\documentclass[10pt,a4paper]{article}
\usepackage[english]{babel}
\usepackage[utf8]{inputenc}
\usepackage{amsfonts,amsbsy,bm,euscript,mathrsfs}
\usepackage{amssymb,stmaryrd,faktor}
\usepackage[tbtags]{amsmath}
\usepackage{bbm}
\usepackage{float}
\usepackage{graphicx}
\usepackage[title,titletoc]{appendix}
\usepackage[bookmarks=true,colorlinks=true,linkcolor=blue,citecolor=blue,urlcolor=blue,bookmarksnumbered]{hyperref}
\usepackage{multicol}

\usepackage{dsfont}
\usepackage{lmodern}
\usepackage{mathrsfs}
\usepackage{mathtools}
\usepackage{bbm}
\usepackage{braket}
\usepackage{slashed}

\usepackage{slashbox} 

\usepackage{arydshln}

\usepackage{tikz-cd}

\numberwithin{equation}{section}

\makeatletter
\renewcommand\section{\@startsection {section}{1}{\z@}
{-3.5ex \@plus -1ex \@minus -.2ex}
{2.3ex \@plus.2ex}
{\normalfont\Large\bfseries}}
\renewcommand\subsection{\@startsection{subsection}{2}{\z@}
{-3.25ex\@plus -1ex \@minus -.2ex}
{1.5ex \@plus.2ex}
{\normalfont\large\bfseries}}
\makeatother

\newcommand{\bea}{\begin{eqnarray}}
\newcommand{\eea}{\end{eqnarray}}

\allowdisplaybreaks 

\usepackage{amsthm}
\newtheorem{theorem}{Theorem}[section]
\newtheorem{corollary}{Corollary}[theorem]
\newtheorem{proposition}[theorem]{Proposition}

\usepackage{blkarray}

\begin{document}

\thispagestyle{empty}
\begin{flushright}\footnotesize\ttfamily
DMUS-MP-21/17
\end{flushright}
\vspace{2em}

\begin{center}

{\Large\bf \vspace{0.2cm}
{\color{black} \large Jordan blocks and the Bethe ansatz I: The eclectic spin chain as a limit}} 
\vspace{1.5cm}

\textrm{Juan Miguel Nieto García\footnote{\texttt{j.nietogarcia@surrey.ac.uk}} and Leander Wyss\footnote{\texttt{l.wyss@surrey.ac.uk}}}

\vspace{2em}

\vspace{1em}
\begingroup\itshape
Department of Mathematics, University of Surrey, Guildford, GU2 7XH, UK
\par\endgroup

\end{center}

\vspace{2em}

\begin{abstract}\noindent 
We present a procedure to extract the generalised eigenvectors of a non-diagonalisable matrix by considering a diagonalisable perturbation of it and computing the non-diagonalisable limit of its eigenvectors. As an example of this process, we compute a subset of the spectrum of the eclectic spin chain by means of the Nested Coordinate Bethe Ansatz. This allows us to show that the Bethe Ansatz of the finitely twisted spin chain contains enough information to reconstruct the generalised eigenvectors of the eclectic spin chain.
\end{abstract}

\newpage

\overfullrule=0pt
\parskip=2pt
\parindent=12pt
\headheight=0.0in \headsep=0.0in \topmargin=0.0in \oddsidemargin=0in

\vspace{-3cm}
\thispagestyle{empty}
\vspace{-1cm}

\tableofcontents

\setcounter{footnote}{0}

\section{Introduction}

Non-Hermitian systems have a plethora of applications in physics, ranging from optics to critical phenomena, even appearing in transport phenomena in biological systems, see \cite{Heiss:2012dx,Ashida:2020dkc} and references therein for particular instances. In spite of that, they are not as widely studied as Hermitian systems due to their complexity. Non-Hermitian systems are even less studied in the context of quantum mechanics, where requiring the Hamiltonian operator to have real eigenvalues and being bounded from below is usually taken care of by demanding its hermiticity. However, being Hermitian is a sufficient condition for a matrix to have real eigenvalues, but it not a necessary condition. In fact, it has been shown that the Dirac-von Neumann axiom regarding the hermiticity of the Hamiltonian can be relaxed to the existence of an anti-linear operator that commutes with the Hamiltonian (for example, $\mathcal{PT}$ symmetry) \cite{Bender:1998ke,Bender:2007nj,Alexandre:2015kra,PT}. Logarithmic Conformal Field Theories deserve a special mention, as it has been proven that the existence of correlation functions with logarithmic singularities is fundamentally tied to the Virasoro operator $L_0$ developing Jordan cells (which cannot happen if $L_0$ is Hermitian) \cite{Gurarie:1993xq}.

Non-Hermitian systems have also been studied in the context of integrability. Some examples are Baxter's $\mathbb{Z}_N$ Hamiltonian \cite{Baxter},  Toda field theories with imaginary coupling \cite{Hollowood:1991vfe}, non-Hermitian extensions of Calogero-Moser-Sutherland models \cite{Fring:2005ys}, the Ising model in the presence of an imaginary magnetic field \cite{Castro-Alvaredo:2009xex}, and quantum many-body models in contact with a Markovian environment \cite{LPP}. More extreme cases where non-hermiticity leads to non-diagonalisability are the loop models studied in \cite{MS} and the $U_q (\mathfrak{sl}(2))$ open spin chain when $q$ becomes a root of unity \cite{Gainutdinov:2016pxy}. In addition, non-diagonalisable integrable structures seem to be relatively prevalent in the classification of $R$-matrices performed in \cite{DeLeeuw:2019gxe}.

A context where non-Hermitian Hamiltonians are especially common is in perturbation theory. Given two Hermitian operators $A$ and $B$, we can study the diagonalisability of the operator $H(\lambda)=A+\lambda B$ and the consequences of analytically continuing the parameter $\lambda$ to complex values. We assume that $A$ and $B$ do not commute, as the problem would be trivial otherwise. Notice that considering complex values of $\lambda$ makes the operator $H(\lambda)$ non-Hermitian, which opens the door to some additional structures. It is well-known that there may exist points $\lambda_0$ such that $H(\lambda_0)$ has \emph{degenerate} eigenvalues. In contrast with Hermitian matrices, non-Hermitian matrices can develop an additional kind of degeneration. In particular, non-Hermitian matrices that depend on continuous parameters may have one or more points where two or more eigenvectors coincide. We call these points \emph{exceptional points}, and we say that eigenvectors \emph{coalesce} to a single eigenvector, following the terminology established in \cite{Kato}. Obviously, both eigenvalues and eigenvectors cannot be analytic functions of $\lambda$ around exceptional points. Rather, they can be described in terms of a fractional power series (related with the number of eigenvectors that coalesce) called \emph{Puiseux series}, see appendix B of \cite{KGTP} for the explicit derivation in the case of Jordan cells of size 2. The properties of these series are discussed in \cite{Knopp}.

In recent years, the interest in non-Hermitian physics has grown in the integrability community. This interest is driven mainly by a particular deformation of $AdS_5 \times S^5$ called \emph{fishnet theory}, proposed in \cite{Gurdogan:2015csr}, which is described in terms of a non-Hermitian Lagrangian. In order to construct the fishnet theory action, we first deform $\mathcal{N}=4$ Super-Yang-Mills (SYM) theory by twisting the product of the fields in the Lagrangian. This theory is called $\gamma_i$-deformation, and it has proven to be integrable \cite{Lunin:2005jy,Frolov:2005dj} but not conformal \cite{Fokken:2013aea}. Strongly twisted models are obtained when one considers the limit of infinite twist, $\gamma_i \rightarrow \infty$, and vanishing coupling, $g\rightarrow 0$, while keeping their product constant, $\gamma_i g=\xi_i$. This limit is very interesting, as the theory recovers the conformal symmetry. The most general of these strongly twisted theories has the following interaction Lagrangian (up to relabelling of the fields)
\begin{align}
 \mathcal{L}_{\text{int}} =&-i N \text{Tr} \left[ \sqrt{\xi_2 \xi_3} (\psi^3 \phi^1 \psi^2 +\bar{\psi}_3 \phi_1^\dagger \bar{\psi}_2 ) + \sqrt{\xi_3 \xi_1} (\psi^1 \phi^2 \psi^3 +\bar{\psi}_1 \phi_2^\dagger \bar{\psi}_3 ) + \sqrt{\xi_1 \xi_2} (\psi^2 \phi^3 \psi^1 +\bar{\psi}_2 \phi_3^\dagger \bar{\psi}_1 ) \right] \notag \\
 &- N \text{Tr} \left[ (\xi_3)^2 \phi_1^\dagger \phi_2^\dagger \phi^1 \phi^2 + (\xi_1)^2 \phi_2^\dagger \phi_3^\dagger \phi^2 \phi^3 + (\xi_2)^2 \phi_3^\dagger \phi_1^\dagger \phi^3 \phi^1 \right]   \ , \label{interactionstrongtwist}
\end{align}
where $\phi^j$ are complex bosonic fields and $\psi^j$ are fermionic fields. The gauge fields and the fourth fermion decouple in this limit. The fishnet theory is a particular case of this strongly twisted Lagrangian where all but one of the deformation parameters $\xi_i$ are set to zero.

Direct inspection of the Lagrangian~(\ref{interactionstrongtwist}) is enough to realise that strongly twisted theories are non-Hermitian. This implies that the dilatation operator associated with the conformal symmetry of the theory is no longer Hermitian for general values of the deformation parameters. It has been observed that not only do the eigenvalues associated with the dilatation operator take complex values, but that it also becomes non-diagonalisable. Similarly to $\mathcal{N}=4$ SYM theory, we can express the action of the dilatation operator on single trace operators at one-loop in terms of the nearest-neighbour Hamiltonian of an effective spin chain \cite{Ipsen:2018fmu}. The effective spin chain that describes the action of the one-loop dilatation operator on single trace operators made of scalars was studied in \cite{Ipsen:2018fmu,StaudacherAhn}. To that end, the authors solved the Bethe equations for the case of finite values of the twist parameter $q_i$ and computed the large twist limit. Although they were able to find the correct number of Bethe roots, they were not able to find all the Bethe vectors. The reason for it was that several Bethe vectors have the same limit at strong twist. They denote those particular configurations as \emph{locked states}, as they all have the relative positions of the excitations fixed. They find that these states are indeed eigenstates of the dilatation operator, but they do \emph{not} exhaust the number of true eigenvectors of the dilatation operator. This disappearance of eigenvectors might seem surprising at first sight, but it is actually the same effect as the coalescence of eigenvectors we described before for perturbation theory. It is a consequence of the non-hermiticity of the one-loop dilatation operator for general values of the deformation parameters $\gamma_i$, and the fact that the limit we are considering is an exceptional point of it.

While finishing up writing this article, \cite{Ahn:2021emp} was published, which extends the counting of the eigenstates done in \cite{StaudacherAhn} to any number of excitations by means of combinatorial arguments. While they perform the counting only on the hyper-eclectic spin chain (that is, effective spin chain associated with $\xi_1=\xi_2=0$), they argue that the counting has to be the same for the eclectic spin chain.  Although this creates a substantial overlap between their results and the ones we present here, the method we used to obtain those results is entirely different. Thus, the two articles complement and reaffirm each other.

Although there is a wide range of theorems on the analytic properties and behaviour of eigenvalues when approaching an exceptional point, we were not able to find as many for properties of the eigenvectors. The most relevant one we found was proposition 4.3 in \cite{Gainutdinov:2016pxy}, where the authors provide a recipe for computing generalised eigenvectors in a particular model with Jordan cells of size two. However, this proposition was crafted specifically for the model they were studying and a situation that only involves Jordan cells of size two.

In this article, we plan to examine in detail what happens to the eigenvectors of a diagonalisable matrix when we approach an exceptional point, and provide a method to extract the generalised eigenvectors at the exceptional point from their analytical properties near this point. Although our results will be general, here we will only apply them to the so called \emph{eclectic spin chain} from \cite{Ipsen:2018fmu,StaudacherAhn}, which describes the one-loop dilatation operator of strongly twisted models. We will construct the explicit form of the eigenvectors of this dilatation operator for generic values of the twist using the Nested Coordinate Bethe Ansatz and check that we can find all the eigenvectors and generalised eigenvectors of a specific subsector using our method. This result also gives us a first-principle explanation for the success of proposition 4.3 in \cite{Gainutdinov:2016pxy}.

The outline of this article is as follows.  In section 2 we review the connection between the strong twist limit of $\mathcal{N}=4$ SYM and the eclectic spin chains introduced in \cite{Ipsen:2018fmu}. In section 3 we recall some results of non-diagonalisable matrices that are relevant to the upcoming computations. In sections 4 and 5 we study how matrices and their eigenvectors behave when we approach the exceptional point of the matrix they are associated with. In section 4 we focus exclusively on the case where all the Jordan cells can be distinguished, while in section 5 we consider the case where two or more Jordan cells can have the same eigenvalue. In both sections we preface our results with the computation of the generalised eigenvectors of some particular subsector of the eclectic spin chain Hamiltonian as examples to introduce our results. In section 6 we apply our results to Bethe vectors of the eclectic spin chain Hamiltonian. In section 7 we present some conclusions and future directions.

\section{Single trace operators in strongly twisted $\mathcal{N}=4$ Super-Yang-Mills and eclectic spin chains}

A very well-known procedure to compute the conformal dimension of operators in $\mathcal{N}=4$ SYM is to construct an effective spin chain whose Hamiltonian acts in the same fashion as the dilatation operator acts on single-trace operators. A detailed review of this method can be found in \cite{Minahan:2010js}.

This procedure can be extended to the $\gamma_i$-deformation, and thus to the strong twist deformation of $\mathcal{N}=4$ SYM. The details of it can be found in \cite{Fokken:2013aea,Ipsen:2018fmu,StaudacherAhn,Fokken:2013mza}. Here we will follow the expressions from \cite{StaudacherAhn}. In this article, we will care only about single-trace operators made of either of the three scalar fields with no derivatives. If we identify the scalar field $\phi^i$ with the spin state $|i\rangle$, the dilatation operator acting on a single trace operator involving a total of $L$ scalar fields takes the form
\begin{equation}
	\mathcal{D}=\mathcal{D}_0 + g^2 \mathbf{\hat{H}}_{\xi_1 , \xi_2 , \xi_3}+\mathcal{O} (g^4) =\mathcal{D}_0 + g^2 \left[\sum_{l=1}^L \hat{\mathbb{P}}^{l,l+1} \right]+\mathcal{O} (g^4) \ ,
\end{equation}
where $\mathcal{D}_0$ is the bare dimension of the operator (which is equal to $L$ for scalar operators), and $\hat{\mathbb{P}}^{a,b}$ is an operator that acts non-trivially on sites $a$ and $b$ as follows
\begin{align}\label{eclecticperm}
\hat{\mathbb{P}}\,|21\rangle &=\xi_3\, |12\rangle \ , & \hat{\mathbb{P}}\,|32\rangle &=\xi_1\, |23\rangle \ , & \hat{\mathbb{P}}\, |13\rangle &=\xi_2\, |31\rangle \ ,
\end{align}
while the remaining matrix elements are zero. In addition, as we are interested in single-trace operators, this means that we will be working with closed spin chains with the periodic identification $L+1\equiv 1$.

If we consider the single trace operator consisting solely of $\phi^1$ as the vacuum state of the effective spin chain, scalars $\phi^2$ behave as right-moving excitations while scalars $\phi^3$ behave as left-moving excitations. We will denote the total number of excitations, i.e. the number of $\phi^2$ and $\phi^3$ fields in the operator, by $M$, and the total number of $\phi^3$ fields in the operator by $K$. Without any loss of generality, we can assume that $K\leq M-K \leq L-M$. If we have operators made of two out of the three fields, i.e. $K=0$, we can diagonalise the Hamiltonian $\mathbf{\hat{H}}$ without any issue. Yet, if the three kinds of fields are present, i.e. $K\neq 0$, we can see that the Hamiltonian cannot be diagonalisable. The simplest route to prove it is to show that applying the Hamiltonian a sufficient number of times on a state containing all three kinds of excitations will make it vanish. In order to see that, we first have to consider that a left-moving excitation and a right-moving excitation will eventually meet after applying the Hamiltonian enough times. Since they cannot reflect off each other, as they would travel in the wrong direction after such an event, they act as impenetrable walls to each other. Thus, excitations start to accumulate at each side of this ``23 wall''. Once all the excitations have accumulated, an additional application of the Hamiltonian with make the state vanish, as no excitation can be moved. Because this statement holds true for any state containing the three kinds of excitations, the Hamiltonian is nilpotent. Following the terminology of \cite{Ipsen:2018fmu}, spin chains having this property are said to have eclectic field content.

Although we can write the Hamiltonian in Jordan normal form using the standard algorithm, we are interested in tacking the problem from a different perspective. Rather than computing the Jordan normal form of the effective Hamiltonian associated with the one-loop dilatation operator of strongly twisted $\mathcal{N}=4$ SYM, we will diagonalise the one associated with the $\gamma_i$-deformation and study how the eigenvalues and eigenvectors behave for large values of the twist. The effective spin chain associated with the $\gamma_i$-deformation corresponds to an integrable deformation of the $\mathfrak{su}(3)$ XXX spin chain termed \emph{twisting}. For the particular case of $\mathfrak{su}(3)$, it admits a total of 3 different twist parameters, $q_i$. Similarly to the untwisted case, the Hamiltonian can be written in terms of a \emph{twisted permutation operator} $\tilde{\mathbb{P}}^{a,b}$
\begin{equation}\label{eq:twistedXXX}
\mathbf{\tilde{H}}_{(q_1,q_2,q_3)}=\sum_{l=1}^L \tilde{\mathbb{P}}^{l,l+1}\ ,
\end{equation}
where $\tilde{\mathbb{P}}^{a,b}$ acts non-trivially on sites $a$ and $b$ as follows
\begin{align}\label{eq:twistedperm}
\tilde{\mathbb{P}}\, |11\rangle &= |11\rangle \ , & \tilde{\mathbb{P}}\,|22\rangle &= |22\rangle \ , & \tilde{\mathbb{P}}\,|33\rangle &= |33\rangle \ , \\ \nonumber
\tilde{\mathbb{P}}\,|12\rangle &=\frac{1}{q_3}\, |21\rangle \ , & \tilde{\mathbb{P}}\,|23\rangle &=\frac{1}{q_1}\, |32\rangle \ , & \tilde{\mathbb{P}}\,|31\rangle &=\frac{1}{q_2}\, |13\rangle \ , \\ \nonumber
\tilde{\mathbb{P}}\,|21\rangle &=q_3\, |12\rangle \ , & \tilde{\mathbb{P}}\,|32\rangle &=q_1\, |23\rangle \ , & \tilde{\mathbb{P}}\, |13\rangle &=q_2\, |31\rangle \ .
\end{align}
Although this Hamiltonian is Hermitian only if the twist parameters are complex phases, it is diagonalisable for generic values of said parameter (see footnote 2 in \cite{StaudacherAhn}).

As we have commented in the introduction, we can recover the eclectic spin chain Hamiltonian by considering the following limit
\begin{equation}
	\mathbf{\hat{H}}_{\xi_1 , \xi_2 , \xi_3}=\lim_{\epsilon \rightarrow 0} \epsilon \mathbf{\tilde{H}}_{(\frac{\xi_1}{\epsilon},\frac{\xi_2}{\epsilon},\frac{\xi_3}{\epsilon})} \ ,
\end{equation}
but we are not able to span the complete Hilbert space with the vectors generated through the limit ${\epsilon \rightarrow 0}$ of the eigenvectors of $\mathbf{\tilde{H}}$. If one computes the Bethe vectors for finite value of the deformation parameters and compute their ${\epsilon \rightarrow 0}$ limit, these vectors coalesce to states of the form
\begin{equation}
	|\psi (k) \rangle=\sum_{l=1}^L e^{2\pi k l i /L} U^l | 1, \dots, 1, 2,\dots , 2, 3, \dots, 3 \rangle
\end{equation}
where $L$ is the total number of sites and $U$ is the operator that shifts the configuration by one site to the left, i.e., $U |n_1, \dots ,n_{L-1}, n_L \rangle=|n_2, \dots ,n_L, n_1 \rangle$. Due to their form, where the excitations cannot move from their relative places, these states have been called \emph{locked states}. Although we can check that they are eigenvectors of the Hamiltonian $\mathbf{\hat{H}}_{\xi_1 , \xi_2 , \xi_3}$, we can also check that these cannot be all the eigenvectors of this Hamiltonian. Thus, naïvely computing the limit at the level of the eigenvectors not only gives us zero information about the generalised eigenvectors of the Hamiltonian, but also it does not provide us with the full set of eigenvectors.

In the following sections we will show that all we have to do to retrieve this information is to compute limits of linear combinations of eigenvectors rather than just computing the limit of eigenvectors.

\section{Non-diagonalisable matrices, generalised eigenvectors and Jordan normal form}

In this section, we review some key concepts about non-diagonalisable matrices. The aim is to introduce the notation we will use throughout this article and highlight some results that will be relevant.

Given a matrix $M$, we say that $\lambda_i$ is an eigenvalue of $M$ with \emph{algebraic multiplicity} $n_i$ if it is a zero of degree $n_i$ of the characteristic polynomial of $M$, that is, if
\begin{equation}
	\det (M-\lambda \mathbb{I} ) \propto (\lambda - \lambda_i)^{n_i} \ ,
\end{equation}
where $\mathbb{I}$ is the identity matrix. We say that $v_i$ is an eigenvector of $M$ associated with the eigenvalue $\lambda_i$ if
\begin{equation}
	M v_i=\lambda_i v_i \ .
\end{equation}
The total number of linearly independent vectors that fulfil this equation, i.e. the dimension of $\text{Ker}(M-\lambda_i I)$, is called \emph{geometric multiplicity} of $\lambda_i$. It is easy to check that the geometric multiplicity is always equal or smaller than the algebraic multiplicity.

A matrix is \emph{non-diagonalisable} or \emph{defective} if the geometric multiplicity of one or more of its eigenvalues is smaller than their algebraic multiplicity. This means that there exist fewer eigenvectors than we expected from the characteristic polynomial. As we cannot diagonalise this kind of matrices, the next best thing we can do in this situation is to ``make up for the missing vectors'' in some fashion. This is done by means of \emph{Jordan chains}. Given a defective eigenvalue $\lambda_i$ and an eigenvector $v_{i,\alpha}^{(1)}$ associated with it, we define the \emph{generalised eigenvector of rank $n$} as the vector fulfilling
\begin{equation}
	(M-\lambda_i \mathbb{I} ) v_{i,\alpha}^{(n)}=v_{i,\alpha}^{(n-1)} \ , \label{geneigenvdefinition}
\end{equation}
where the index $\alpha$ labels the possible geometric multiplicity of the eigenvalue $\lambda_i$. For simplicity, we will usually drop the $(1)$ superindex for eigenvectors. It is trivial to prove that the generalised eigenvectors fulfil $(M-\lambda_i \mathbb{I} )^n v_{i,\alpha}^{(n)}=0$. In addition, the vectors that form a Jordan chain are linearly independent but not necessarily orthogonal.

Thanks to these additional linearly independent vectors, we can write a similarity transformation that changes a given square matrix into a block diagonal matrix of the form
\begin{equation}
	M=\begin{pmatrix}
	J_1 & 0 & 0 &   \dots \\
	0 & J_2 & 0 &  \dots \\
	0 & 0 & J_3  & \dots \\
	\vdots & \vdots & \vdots & \ddots
	\end{pmatrix} \ ,  \text{ with each block being of the form } J_i=\begin{pmatrix}
	\lambda_i & 1 & 0 & 0 &  \dots \\
	0 & \lambda_i & 1 & 0 & \dots \\
	0 & 0 & \lambda_i & 1 & \dots \\
	0 & 0 & 0 & \lambda_i & \dots \\
	\vdots & \vdots & \vdots & \vdots & \ddots
	\end{pmatrix} \ .
\end{equation}
This matrix is called the \emph{Jordan normal form} of $M$, and each block is called \emph{Jordan block} or \emph{Jordan cell}.\footnote{Jordan normal forms can also be constructed for matrices of generic dimension $N \times M$, however, we will only consider square matrices in this article.} The size of each Jordan cell is equal to the number of vectors that form the Jordan chain associated with it.

Notice that, similarly to what happens in diagonalisable matrices with degenerated eigenvalues, the linear combinations of two generalised eigenvectors of rank $p$ associated with the same eigenvalue is also a generalised eigenvector of rank $p$, but the generalised eigenvector of rank $p-1$ it is associated with changes. On the other hand, the linear combination of a generalised eigenvector of rank $p$ and a generalised eigenvector of rank $1$ associated with the same eigenvalue is also a generalised eigenvector of rank $p$ that is associated with the same generalised eigenvector of rank $p-1$ as the original.

Among the properties of Jordan chains, in this article we will make an extensive use of the following theorem
\begin{theorem}
The number of generalised eigenvectors of rank $l$ is never larger than the number of true eigenvectors, i.e., \emph{$\text{dim}\left\{ \text{Ker}\, (M-\lambda \mathbb{I})^l \right\} - \text{dim}\left\{ \text{Ker}\, (M-\lambda \mathbb{I})^{l-1} \right\} \leq \text{dim}\left\{ \text{Ker}\, (M-\lambda \mathbb{I}) \right\} $}.

In fact, \emph{$\text{dim}\left\{ \text{Ker}\, (M-\lambda \mathbb{I})^l \right\} - \text{dim}\left\{ \text{Ker}\, (M-\lambda \mathbb{I})^{l-1} \right\} \leq \text{dim}\left\{ \text{Ker}\, (M-\lambda \mathbb{I})^n \right\} - \text{dim}\left\{ \text{Ker}\, (M-\lambda \mathbb{I})^{n-1} \right\} $} for any two $n<l$. \label{dimker}
\end{theorem}
\begin{proof}
For any vector $w\in \text{Ker}\, (M-\lambda \mathbb{I})^l$, it is obvious that $(M-\lambda \mathbb{I})^{l-1} w\in \text{Ker}\, (M-\lambda \mathbb{I})$ by definition. However, some of the vectors in $\text{Ker}\, (M-\lambda \mathbb{I})^l$ only fulfil this property because $(M-\lambda \mathbb{I})^{l-1} w=0$ and should be disregarded. This means that the vector space $\text{Ker}\, (M-\lambda \mathbb{I})^l / \text{Ker}\, (M-\lambda \mathbb{I})^{l-1}$ is isomorphic to a subset of $(M-\lambda \mathbb{I}) $, which gives us the previous result by comparing their dimensions.

The second part of the theorem can be proven similarly.
\end{proof}

We should stress again that, despite being linearly independent, the generalised eigenvectors of a non-Hermitian matrix do not have to be orthogonal. We can recover a notion of orthogonality of eigenvectors, but it requires us to consider the generalised eigenvectors of the matrix $M^\dagger$. In the literature, regular eigenvectors are usually called \emph{right eigenvectors} of $M$ while the eigenvectors of $M^\dagger$ are called \emph{left eigenvectors} of $M$. Similar names exist for generalised eigenvectors. This terminology comes from the relation
\begin{equation}
	M^\dagger \hat{v}_{i,\alpha} = \lambda^*_{i} \hat{v}_{i,\alpha} \Longleftrightarrow \hat{v}_{i,\alpha}^\dagger M=\lambda_{i} \hat{v}_{i,\alpha}^\dagger \ .
\end{equation}
Left and right eigenvectors coincide for Hermitian matrices, but that is not the case for non-Hermitian ones. In fact, the left and right eigenvectors associated with the same Jordan block are orthogonal. This property is usually called \emph{self-orthogonality}, and it can be used to find exceptional points.

By introducing the concept of left eigenvectors, we can generalise the orthogonality relations that hold for Hermitian matrices to the following \emph{biorthogonality} relations
\begin{theorem}
	The left and right generalised eigenvectors of a square matrix can be chosen such that
	\begin{equation}
		(\hat{v}_{i,\alpha}^{(n_{i,\alpha}+1-p)})^\dagger \cdot v_{j,\beta}^{(q)}=\delta_{ij} \delta_{\alpha \beta} \delta_{pq} \ ,
	\end{equation}
	where $n_{j,\alpha}$ is the size of the Jordan block associated with the $\alpha$-th geometric multiplicity of $\lambda_j$.
\end{theorem}
\begin{proof}
On the one hand, we have that a left and a right eigenvector associated with two different eigenvalues are orthogonal. The proof is similar to the proof of the orthogonality of eigenvectors in Hermitian matrices
\begin{displaymath}
0=\hat{v}_{i,\alpha}^\dagger M  v_{j,\beta} -\hat{v}_{i,\alpha}^\dagger M v_{j,\beta}= (\lambda_i - \lambda_j) \hat{v}_{i,\alpha}^\dagger \cdot v_{j,\beta} \ ,
\end{displaymath}
as we are considering $\lambda_i \neq \lambda_j$, the only possibility is that $\hat{v}_{i,\alpha} $ and $ v_{j,\beta}$ are orthogonal.

Now, we need to generalise this result to encompass generalised eigenvectors. To that end, we need the following two relations
\begin{gather*}
	(\lambda_i - \lambda_j) (\hat{v}_{i,\alpha}^{(1)})^\dagger  \cdot v_{j,\beta}^{(p)}= (\hat{v}_{i,\alpha}^{(1)})^\dagger (M-\lambda_j \mathbb{I} ) v_{j,\beta}^{(p)}=(\hat{v}_{i,\alpha}^{(1)})^\dagger  \cdot v_{j,\beta}^{(p-1)} \ , \\
	(\hat{v}_{i,\alpha}^{(q-1)})^\dagger \cdot v_{j,\beta}^{(p)}+ (\lambda_i - \lambda_j) (\hat{v}_{i,\alpha}^{(q)})^\dagger \cdot v_{j,\beta}^{(p)}= (\hat{v}_{i,\alpha}^{(q)})^\dagger (M-\lambda_j \mathbb{I} ) v_{j,\beta}^{(p)}=(\hat{v}_{i,\alpha}^{(q)})^\dagger \cdot v_{j,\beta}^{(p-1)} \ .
\end{gather*}
With these relations, we can show by induction that a set of left generalised eigenvectors are orthogonal to a set of right generalised eigenvectors if they are associated with different eigenvalues $\lambda_i \neq \lambda_j$.

The next step is to consider what happens with generalised eigenvectors associated with the same eigenvalue. For that, we have to consider the following relation
\begin{displaymath}
	(\hat{v}_{j,\alpha}^{(p)})^\dagger (M-\lambda_j \mathbb{I})^{k+l} v_{j,\beta}^{(q)}=(\hat{v}_{j,\alpha}^{(p)})^\dagger  \cdot v_{j,\beta}^{(q-k-l)}=(\hat{v}_{j,\alpha}^{(p-k-l)})^\dagger  \cdot v_{j,\beta}^{(q)}=(\hat{v}_{j,\alpha}^{(p-k)})^\dagger  \cdot v_{j,\beta}^{(q-l)} \ .
\end{displaymath}
If we consider the case of $k=p-1$ and $l>0$, we have that an eigenvector is orthogonal to all the generalised eigenvectors except for the ones of maximal rank. If the geometric multiplicity of the eigenvalue is one, we are done: we have proven that the left eigenvector $\hat{v}_j^{(1)}$ is orthogonal to all right generalised eigenvectors except the one of maximal rank $\hat{v}_j^{(n_j)}$. As this set of generalised eigenvectors span the full vector space and $\hat{v}_j^{(1)}$ is non-vanishing, $(\hat{v}_j^{(1)})^\dagger {v}_j^{(n_j)}\neq 0$, and can be set to 1. Furthermore, by the above equation, we can extend this result to $(\hat{v}_j^{(1+l)})^\dagger {v}_j^{(n_j-l)}=1$, with the other possible combination giving us zero. If the geometric multiplicity is different from one, the proof is similar but requires a process akin to Gram-Schmidt orthogonalisation that takes advantage of the fact that the linear combination of a generalised eigenvector of rank $p$ and a generalised eigenvector of rank $1$ associated with the same eigenvalue is also a generalised eigenvector of rank $p$.
\end{proof}

We should emphasise that biorthogonality does not imply that $\hat{v}^{(n_{i,\alpha}+1-p)}_{i,\alpha}=v^{(p)}_{i,\alpha}$ due to the generalised eigenvectors not being an orthogonal basis. However, there are two important exceptions
\begin{corollary} \label{lowesttohighest}
	The following equalities between left and right generalised eigenvectors hold: $\hat{v}^{(n_{i,\alpha})}_{i,\alpha}=v^{(1)}_{i,\alpha}$ and $\hat{v}^{(1)}_{i,\alpha}=v^{(n_{i,\alpha})}_{i,\alpha}$.
\end{corollary}
\begin{proof}
	Because the generalised eigenvector form a basis, we can expand a left generalised eigenvector in terms of right generalised eigenvectors
	\begin{displaymath}
		\hat{v}_{i,\alpha}^{(n_{i,\alpha}+1-p)}=\sum_{k,\gamma,l} m_{i,k,\alpha,\gamma,p,l} v^{(l)}_{k,\gamma} \ .
	\end{displaymath}
	If we apply biorthogonality, we find that the coefficients $m$ have so satisfy the relation
	\begin{displaymath}
		\delta_{ij} \delta_{\alpha \beta} \delta_{pq}= \sum_{k,\gamma,l} m_{i,k,\alpha,\gamma,p,l} (v^{(l)}_{k,\gamma} )^\dagger v^{(p)}_{j,\beta} \ .
	\end{displaymath}
	As the generalised eigenvectors are not orthogonal, we cannot set $m$ to a Kronecker delta.
	
	The only exception to the above statement are the generalised eigenvectors of rank 1, that is, the eigenvectors. As the linear combination of a generalised eigenvector of rank $q$ and an eigenvector is a generalised eigenvector of rank $q$ (notice that this will require to linearly combine the generalised eigenvector of rank $q+1$, if it exists, with a generalised eigenvector of rank 2, so equation (\ref{geneigenvdefinition}) keeps holding for this generalised eigenvector), we can linearly combine generalised eigenvectors in such a way that they are orthogonal to all the generalised eigenvectors of rank 1. Using this result, we can set $m$ to a Kronecker delta when $p=1$, meaning that $\hat{v}^{(n_{i,\alpha})}_{i,\alpha}=v^{(1)}_{i,\alpha}$.
	
	Expanding a right generalised eigenvector in terms of left generalised eigenvectors and repeating the same steps gives us the other relation.
\end{proof}

\begin{corollary}
	Any square complex matrix $M$ can be expressed as
\begin{equation}
	M=\sum_{j=1}^J \sum_{\alpha=1}^{m^g_j} \left( \sum_{p=1}^{n_{j\alpha}} \lambda_jv_{j,\alpha}^{(p)}  \left(\hat{v}_{j,\alpha}^{(n_{j\alpha}+1-p)}\right)^\dagger +\sum_{p=1}^{n_{j\alpha}-1} v_{j\alpha}^{(p+1)} \left(\hat{v}_{j,n_{\alpha}}^{(n_{j\alpha}+1-p)}\right)^\dagger \right)
\end{equation}
where $J$ is the total number of different eigenvalues, $m^g_j$ is the geometric multiplicity of the eigenvalue $\lambda_j$ and $n_{j\alpha}$ is the size of the Jordan cell associated with the $\alpha$-th eigenvector associated with $\lambda_j$.
\end{corollary}
The proof of this corollary is immediate using the definition of generalised eigenvectors and the biorthogonality property.

In the context of non-Hermitian matrices there exist two different definitions for the norm of an eigenvector. On the one hand, we can define it using the usual inner product in a complex vector space, that is, $|v_{i,\alpha}^{(p)}|=\left( v_{i,\alpha}^{(p)} \right)^\dagger \cdot v_{i,\alpha}^{(p)}$. This is the definition we will be using in this article. On the other hand, we can define it using left and right eigenvectors, $|v_{i,\alpha}^{(p)}|_{lr}=\left( \hat{v}_{i,\alpha}^{(p)}\right)^\dagger \cdot v_{i,\alpha}^{(p)}$. Although we will not use this definition in the article, it is extremely useful to find exceptional points of non-Hermitian matrices, as this norm is non-vanishing for diagonalisable matrices but it vanishes for non-diagonalisable matrices due to self-orthogonality.

\section{Exceptional point with distinguishable Jordan blocks}

In this section, we will consider the case where several eigenvectors coalesce in such a way that the Hamiltonian at the exceptional point is composed of Jordan blocks with pairwise different eigenvalues. We will show first the particular case of the Hamiltonian associated with an eclectic spin chain of length 3 and two excitations with different flavour. All the Jordan blocks in this case are of size two, so a slight modification of proposition 4.3 from \cite{Gainutdinov:2016pxy} is enough to compute everything. After this example, we will present a method that works for any defective matrix with  Jordan cells of any size and pairwise different eigenvalues.

\subsection{A specific example: three-sites chain with three different excitations}\label{3sites}

As a warm-up example, we will consider the simplest case where the strongly twisted limit of the Hamiltonian (\ref{eq:twistedXXX}) becomes non-diagonalisable. This corresponds to the spin chain of length 3 with two excitations of different flavour, i.e. $L=3$, $M=2$, $K=1$. As we commented in the previous section, the plan is to compute the eigenvectors of the twisted $\mathfrak{su}(3)$ Hamiltonian $\mathbf{\tilde{H}}_{(q_1,q_2,q_3)}$ and study how they behave at large values of the twist parameters.

The general wavefunction on this sector can be written as
\begin{equation}
	\psi=\sum_{\sigma\in \mathcal{S}_3} \psi(\sigma) |\sigma (1) \sigma (2) \sigma (3)\rangle \ ,
\end{equation}
where $\mathcal{S}_3$ is the symmetric group over a set of three elements. The twisted Hamiltonian (\ref{eq:twistedXXX}) transforms even elements of $\mathcal{S}_3$ into odd elements of $\mathcal{S}_3$. Thus, it can be represented as the following anti-block-diagonal matrix
\begin{equation}
	\mathbf{\tilde{H}}_{(q_1,q_2,q_3)}=\left( \begin{array}{ccc;{2pt/2pt}ccc}
		0 & 0 & 0 & q_3 & q_1 & q_2 \\
		0 & 0 & 0 & q_2 & q_3 & q_1 \\
		0 & 0 & 0 & q_1 & q_2 & q_3 \\ \hdashline[2pt/2pt]
		\frac{1}{q_3} & \frac{1}{q_2} & \frac{1}{q_1} & 0 & 0 & 0 \\
		\frac{1}{q_1} & \frac{1}{q_3} & \frac{1}{q_2} & 0 & 0 & 0 \\
		\frac{1}{q_2} & \frac{1}{q_1} & \frac{1}{q_3} & 0 & 0 & 0 \\
	\end{array} \right) \ .
\end{equation}
This matrix can be diagonalised for generic values of the twist parameters, with eigenvalues
\begin{align}
	\lambda_1^\pm &= \pm \sqrt{\sum_{i,j} \frac{q_i}{q_j}} \ , \\
	\lambda_2^\pm &= \pm \frac{1}{\sqrt{2}} \sqrt{9+\frac{(q_1 - q_2)(q_2-q_3)(q_3-q_1)}{q_1 q_2 q_3}-\sum_{i,j} \frac{q_i}{q_j}} \ , \\
	\lambda_3^\pm &= \pm \frac{1}{\sqrt{2}} \sqrt{9-\frac{(q_1 - q_2)(q_2-q_3)(q_3-q_1)}{q_1 q_2 q_3}-\sum_{i,j} \frac{q_i}{q_j}} \ .
\end{align}
Notice that the sums include the case $i=j$. Due to the rescaling needed to keep the Hamiltonian finite in the strongly twisted limit, these three eigenvalues vanish in said limit, $E_i^\pm=\lim_{\epsilon\rightarrow 0} \epsilon \lambda_i^\pm = 0$.

Let us consider the eigenvectors associated with the eigenvalues $\lambda_1^\pm$, as they are the simplest ones. They take the following form when properly normalised
\begin{equation}
	v_1^\pm=\sqrt{\frac{(q_1 + q_2 + q_3)^2}{2(q_1 + q_2 + q_3)^2+ 3(\lambda_1^\pm)^2}} \left( \frac{q_1 + q_2 + q_3}{\lambda_1^\pm} , \frac{q_1 + q_2 + q_3}{\lambda_1^\pm} , \frac{q_1 + q_2 + q_3}{\lambda_1^\pm}, 1, 1, 1 \right) \ ,
\end{equation}
with all the $\pm$ signs correlated. Substituting $q_i \rightarrow \frac{\xi_i}{\epsilon}$, it is easy to check that both vectors become the same up to a sign in the limit of large twist
\begin{equation}
	\lim_{\epsilon\rightarrow 0} v_1^\pm =\frac{\pm 1}{\sqrt{3}} \left( 1,1,1,0,0,0 \right) \ .
\end{equation}
A similar situation occurs for the other four eigenvectors: they become pairwise proportional in the limit of vanishing $\epsilon$. This happens because the Hamiltonian becomes strictly upper triangular in this limit and, thus, defective. In fact, the Hamiltonian becomes similar to a matrix composed of three Jordan blocks with eigenvalues equal to zero, and the three vectors we obtain through this procedure coincide with the true eigenvectors of this matrix.

After taking the large twist limit, we lose access to half of our Hilbert space, as this example is a particular case of the coalescence of eigenvectors shown in \cite{StaudacherAhn}. We will argue below that this is not strange because the three remaining vectors of the Hilbert space are not eigenvectors of the strongly twisted Hamiltonian, but generalised eigenvectors of rank 2. Then, this begs the question whether we can extract these generalised eigenvectors at large twist from the eigenvectors at finite twist.

A similar situation was observed in the context of open spin chains with $U_q (\mathfrak{sl}(2))$ symmetry, which become defective when $q$ becomes a root of unity. In that case, it was shown that the generalised eigenvectors can be recovered \cite{Gainutdinov:2016pxy}. To do so, the authors did not only have to consider the limit of the eigenvectors of the Hamiltonian, but also the limit of the eigenvectors after subtracting their leading contribution and multiplying by the appropriate factor to make the limit non-vanishing. Although the method proposed in that article works for the case of Jordan blocks of size two, it has some issues if we try to extend it for larger Jordan bocks. We will show that the correct method involves the difference of two eigenvectors instead. For the example we are considering here, let us take again the two simplest eigenvectors $v_1^\pm$ and compute the (normalised) linear combination
\begin{equation}
	\frac{v_1^+ + v_1^-}{|v_1^+ + v_1^-|}=\frac{1}{\sqrt{3}} \left( 0,0,0,1,1,1 \right) \ .
\end{equation}
Although in this particular case all dependence on the twist parameters disappears after normalising, this will not happen in general. We can check that the vector we have obtained through this procedure is linearly independent of the eigenvector $(1,1,1,0,0,0)$ we obtained previously. Not only that, but we can also check that it is the generalised eigenvector of rank 2 of the strongly twisted Hamiltonian associated with it, up to a scalar factor. If similar computations are done for the remaining coalescing pairs of eigenvectors, we will find the two remaining generalised eigenvectors.

As we said above, in this section we will focus on the case of \emph{distinguishable Jordan blocks}. By that, we mean that all the eigenvalues have geometric multiplicity one. We should point out that the example we have studied might not seem of this type, as the three Jordan blocks have vanishing energies, but it is. Although we cannot distinguish the three Jordan blocks through the energy, they are associated with different values of the total (twisted) momentum operator. The exponential of the momentum operator is the shift operator $U$ we introduced previously. We can check that it takes different values in the three different Jordan blocks. In addition, $U$ commutes with the Hamiltonian $\mathbf{\hat{H}}$, so vectors with different momenta cannot mix and the three Jordan blocks are effectively distinguishable.

\subsection{The defective limit of a diagonalisable matrix: distinguishable Jordan blocks}
\label{distinguishableJB}

Although the observation that generalised eigenvectors could be obtained as the next-to-leading order correction of coalescing eigenvectors when we compute the defective limit of a matrix is not completely new, there does not seem to be a first-principle justification for it in the literature. In the following lines we plan to show that the reason why the (apparently ad hoc) method from  \cite{Gainutdinov:2016pxy} works is not accidental, but it is based on linear algebra and analyticity results.


Let us consider a $N\times N$ complex matrix that depends on a parameter $\epsilon$, $M(\epsilon)$. We will assume that such a matrix is diagonalisable for all values of $\epsilon$ except for a finite set of values. In particular, we will assume that the matrix is not diagonalisable for $\epsilon=0$, where it becomes a defective matrix similar to a single Jordan block of size $N$. Outside its exceptional points, we will denote the eigenvalues and eigenvectors of $M(\epsilon )$ by $\lambda_i$ and $v_i$ respectively. Although they depend explicitly on $\epsilon$, we will not explicitly write that dependence most of the time as it will clutter our expressions.

Instead of computing the eigenvalue, eigenvector and generalised eigenvectors of $M(0)$ directly, we want to see if we can obtain the same information by computing the $\epsilon\rightarrow 0$ limit of the eigenvalues and eigenvectors of the diagonalisable matrix $M(\epsilon)$. It is immediate to see that all the eigenvalues have to become equal, because $M(0)$ is similar to a single Jordan block.

It is immediately clear that a non-diagonalisable matrix has to have all its eigenvalues equal. As we are assuming that $M(0)$ is similar to a single Jordan block, we have that
\begin{equation}
\lim_{\epsilon\rightarrow 0} \lambda_i = \lim_{\epsilon\rightarrow 0} \lambda_j \ ,
\end{equation}
for any two eigenvalues $\lambda_i$ and $\lambda_j$ of $M(\epsilon)$.

The fact that all eigenvalues approach the same value as the matrix becomes defective has important implications for the eigenvectors. If we consider the limit of the eigenvector equation, we find that
\begin{equation}
	0=\lim_{\epsilon\rightarrow 0} \left[(M(\epsilon)-\lambda_i \mathbb{I}) v_i\right]= \left[\lim_{\epsilon\rightarrow 0}(M(\epsilon)-\lambda_i \mathbb{I})\right]  \left[ \lim_{\epsilon\rightarrow 0}v_i\right]= \left[M(0)- \mathbb{I}\lim_{\epsilon\rightarrow 0} \lambda_i\right]  \left[ \lim_{\epsilon\rightarrow 0}v_i\right] \ . \label{important}
\end{equation}
As all the eigenvalues degenerate, the matrix $[M(0)- \mathbb{I}\lim_{\epsilon\rightarrow 0} \lambda_i]$ is independent of the index $i$. In addition, because $M(0)$ is similar to a Jordan block and $\lim_{\epsilon\rightarrow 0} \lambda_i$ is the eigenvalue associated with it, then $\lim_{\epsilon\rightarrow 0}v_i$ has to be the true eigenvector of the Jordan block. As there is just one true eigenvector per Jordan block, all the eigenvectors of the diagonalisable matrix $M(\epsilon)$ have to approach the same vector (up to a possible normalisation factor) when $\epsilon$ approaches zero. Therefore, the ``collapse of Bethe vectors'' described in \cite{StaudacherAhn} is just a natural consequence of the matrix becoming of Jordan form. Nevertheless, we should stress that it is only a sufficient condition. Although this has no consequence now, because the kernel of $[M(0)- \mathbb{I}\lim_{\epsilon\rightarrow 0} \lambda_i ]$ is one-dimensional, this will pose a problem when two or more Jordan blocks have the same eigenvalue.

After obtaining information about eigenvalues and eigenvectors, we may wonder if it is also possible to find the generalised eigenvectors by computing limits of eigenvectors. To answer this question, we first have to realise that a linear combination of eigenvectors fulfils the following equation
\begin{equation}
	[M(\epsilon) - \lambda_1 \mathbb{I} ][M(\epsilon) - \lambda_2 \mathbb{I} ] (\alpha_1 \, v_1 + \alpha_2 \, v_2)=0 \ ,
\end{equation}
for any constants $\alpha_1$ and $\alpha_2$, which may depend on $\epsilon$. This equation can be immediately generalised to more general linear combinations. Although it may seem weird to consider linear combinations of eigenvectors associated with different eigenvalues, we need to have in mind that said eigenvalues become identical in the limit we are interested in. In fact, when we compute the limit of the above equation, we find
\begin{multline}
	\left[M(0)- \mathbb{I}\lim_{\epsilon\rightarrow 0} \lambda_1\right] \left[M(0)- \mathbb{I}\lim_{\epsilon\rightarrow 0} \lambda_2\right]  \left[ \lim_{\epsilon\rightarrow 0} (\alpha_1 \, v_1 + \alpha_2 \, v_2) \right]= \\ =\left[M(0)- \mathbb{I}\lim_{\epsilon\rightarrow 0} \lambda_i \right]^2  \left[ \lim_{\epsilon\rightarrow 0} (\alpha_1 \, v_1 + \alpha_2 \, v_2) \right]=0 \ , \label{GEEquation}
\end{multline}
where $\lambda_i$ is any of the eigenvalues of $M(\epsilon)$. This equation implies that the limit of the linear combination contains information regarding the generalised eigenvector of rank 1 (the true eigenvector) and the generalised eigenvector of rank 2. That is, we need to consider the limit of linear combinations of eigenvectors whose associated eigenvalues approach the same eigenvalue if we want to find information about generalised eigenvectors.

Despite equation~(\ref{GEEquation}) assuring us that the linear combination of eigenvectors contains information about the generalised eigenvector of rank 2, this information is not immediately accessible. If we consider general values of the coefficients $\alpha_i$, the linear combination is just a generalised eigenvector of rank 1. This happens because
\begin{equation}
	\lim_{\epsilon\rightarrow 0} [M(\epsilon) - \lambda_1 \mathbb{I} ] (\alpha_1 \, v_1 + \alpha_2 \, v_2)= \lim_{\epsilon\rightarrow 0} \alpha_2 \, (\lambda_2 - \lambda_1) \, v_2= 0 \ , \label{lceigenvector}
\end{equation}
where the last equality comes from the two eigenvalues coinciding at the exceptional point. Thus, the limit of the linear combination is proportional to the true eigenvector of $M(0)$. If we want to find the generalised eigenvector of rank 2, we need to find vectors that are part of Ker$\{\left[M(0)- \mathbb{I}\lim_{\epsilon\rightarrow 0} \lambda_i\right]^2 \}$ but not part of Ker$\{\left[M(0)- \mathbb{I}\lim_{\epsilon\rightarrow 0} \lambda_i\right] \}$. This means that we need to make the right-hand side of (\ref{lceigenvector}) non-vanishing, which can be accomplished by considering a coefficient $\alpha_2$ that diverges as $(\lambda_2 - \lambda_1)^{-1}$. Repeating the argument with $[M(\epsilon) - \lambda_2 \mathbb{I} ]$ we find that $\alpha_1$ has to diverge in the same fashion. Despite that, using $\alpha_i^{-1} \propto v_1 - v_2$ is equally effective, as both quantities vanish at the same rate,\footnote{In fact, by solving the eigenvalue equations for small but non-zero values of $\epsilon$, one can check that $[v_i - \lim_{\epsilon\rightarrow 0} v_i]\propto \lambda_i -\lim_{\epsilon\rightarrow 0} \lambda_i \propto \epsilon^{1/N}$. See, for example, appendix B of \cite{KGTP} for the case of a Jordan block of size 2. This corresponds to the first term of the Puiseux series we commented in the introduction.} and will prove more useful for later generalisation. For this vector to be well-defined, we need that $\lim_{\epsilon\rightarrow 0} (\lambda_2 - \lambda_1) (\alpha_1 \, v_1 + \alpha_2 \, v_2)=0$.

Inverting the argument, if we find two finite coefficients $\alpha_1$ and $\alpha_2$ such that they do not vanish in the $\epsilon\rightarrow 0 $ limit and $\lim_{\epsilon\rightarrow 0} (\alpha_1 \, v_1 + \alpha_2 \, v_2)=0$, then $\lim_{\epsilon\rightarrow 0} \frac{\alpha_1 \, v_1 + \alpha_2 \, v_2}{\lambda_1 - \lambda_2}$ gives us the generalised eigenvector of rank 2.

Furthermore, we have to take into account that, even if two vectors that coalesce are normalised to the same value, they can still differ by a phase factor,  i.e.,  $\lim_{\epsilon\rightarrow 0}v_i=e^{i\phi} \lim_{\epsilon\rightarrow 0}v_j$. Consider then the following two linear combinations
\begin{equation}
w_{ij}^\pm=\frac{ v_i \pm \beta_{ji}v_j}{| v_i \pm \beta_{ji}v_j|} \qquad \text{with} \qquad \beta_{ji}= v_j^\dagger \cdot v_i \ , \label{recipeextra}
\end{equation}
where $|v|^2=v^\dagger \cdot v$ is the usual vector norm. We can check that $\lim_{\epsilon\rightarrow 0} w_{ij}^+$ is equal (up to a phase factor) to $\lim_{\epsilon\rightarrow 0} v_i$, while $\lim_{\epsilon\rightarrow 0} w_{ij}^-$ becomes a $\frac{0}{0}$ indeterminate form. Thus, the latter provides the generalised eigenvector of rank 2 when we apply L'Hôpital's rule. This result provides a first-principle explanation for the proposition 4.3 in \cite{Gainutdinov:2016pxy}.

We should stress that, by applying the same reasoning we used for the collapse of the eigenvectors, every $\lim_{\epsilon\rightarrow 0} w_{ij}^-$ is equal to the generalised eigenvector of rank two regardless of the vectors $v_i$ and $v_j$ we used. This is consistent with theorem \ref{dimker}, which implies that having exactly one true eigenvector means that we can have one generalised eigenvector of rank 2 at most.

From this result it is evident how to continue forward: if we want to find generalised eigenvectors of higher rank, we need to consider linear combinations involving a higher number of eigenvectors of $M(\epsilon)$. This is because the limit of the linear combination of $n$ eigenvectors will fulfil the equation
\begin{equation}
	\left[M(0)- \mathbb{I}\lim_{\epsilon\rightarrow 0} \lambda_i\right]^n  \left[ \lim_{\epsilon\rightarrow 0} \left( \sum_{i=1}^n \alpha_i \, v_i \right) \right]=0 \ ,
\end{equation}
for any constants $\alpha_i$, meaning that the limit of these kinds of linear combinations has information on generalised eigenvectors up to rank $n$. In addition, the previous method for obtaining the generalised eigenvector of rank 2 can be immediately generalised by substituting the vectors $v_i$ by the appropriate ones
\begin{equation}
	w_{ij}^{(n)}=\frac{ w^{(n-1)}_{ji} - \beta^{(n-1)}_{kj} w^{(n-1)}_{ki}}{|w^{(n-1)}_{ji} - \beta^{(n-1)}_{kj} w^{(n-1)}_{ki}|} \qquad \text{with} \qquad \beta^{(n-1)}_{kj}= (w^{(n-1)}_{ji})^\dagger \cdot w^{(n-1)}_{ki}  \quad \text{and} \quad w^{(0)}_{ij}=v_i\ . \label{recipe}
\end{equation}
We should be careful because the $\epsilon\rightarrow 0$ limit of $w_{ij}^{(n)}$ does not give us the generalised eigenvector of rank $n+1$, but a linear combination of generalised eigenvectors up to rank $n+1$ (with a non-vanishing coefficient for the generalised eigenvector of rank $n+1$). The reason behind it is that we are actually solving the weaker condition $\left[M(0)- \mathbb{I}\lim_{\epsilon\rightarrow 0} \lambda_i\right]^{n+1} w^{(n)}=0$ instead of the stronger one $\left[M(0)- \mathbb{I}\lim_{\epsilon\rightarrow 0} \lambda_i\right] w^{(n)}=w^{(n-1)}$. In addition, notice that we can only repeat this process up to $w_{ij}^{(N-1)}$ because $M(\epsilon)$ only has $N$ independent eigenvectors. This is expected, as $M(0)$ can have only $N$ generalised eigenvectors.

We still have to answer the following three questions regarding this recipe for reconstructing the Jordan block: Are all the vectors we get non-zero? Are all the vectors we get independent? Do we get all the generalised eigenvectors?

\begin{proposition}
For every $n\leq N-1$, every $\lim_{\epsilon \rightarrow 0} w^{(n)}_{ij}$ is neither zero nor diverges and it is orthogonal to any $\lim_{\epsilon \rightarrow 0} w^{(m)}_{ij}$ with $m\neq n$.
\end{proposition}
\begin{proof}
The proof for the first statement goes as follows: The vectors $w^{(n)}$ are normalised for any value of $\epsilon$. Because addition and squaring commute with computing limits, the limit of the norm is equal to the norm of the limit. Thus, $\lim_{\epsilon \rightarrow 0} w^{(n)}_{ij}$ cannot be zero because $w^{(n)}_{ij}$ is normalised to $1$ for any value of $\epsilon$.

Let us now prove the orthogonality statement, starting with the generalised eigenvectors of rank $1$ and $2$. For that, we need to compute the scalar product between the limits of $w^{(1)}_{ij}$ and $v_j$, $( \lim_{\epsilon\rightarrow 0} v_j)^\dagger (\lim_{\epsilon\rightarrow 0} w^{(1)}_{ij})$. Instead of this limit, it is simpler to compute instead $\lim_{\epsilon\rightarrow 0} \lim_{\epsilon '\rightarrow \epsilon} ( v_j (\epsilon) )^\dagger ( w^{(1)}_{ij}(\epsilon '))$, which should coincide if the limit exists. As the second limit is smooth, we can prove that the limits of $w^{(1)}_{ij}$ and $v_j$ are orthogonal if the vectors $w^{(1)}_{ij}(\epsilon )$ and $v_j (\epsilon )$ are orthogonal for any $\epsilon$. Using (\ref{recipe}), we get
\begin{displaymath}
	\lim_{\epsilon '\rightarrow \epsilon} ( v_j (\epsilon) )^\dagger ( w^{(1)}_{ij}(\epsilon '))= ( v_j (\epsilon) )^\dagger ( w^{(1)}_{ij}(\epsilon) )=\frac{ (v_j^\dagger \cdot  v_i) - \beta_{ji} (v_j^\dagger \cdot  v_j)}{|v_i - \beta_{ji} v_j|} \ .
\end{displaymath}
If we now substitute the definition of $\beta_{ji}$ and use the fact that the eigenvectors are normalised, the numerator vanishes identically. Thus, the limit of the scalar product vanishes for any $\epsilon$, including $\epsilon=0$.

By similar means, we can prove the orthogonality between the vectors $w^{(n)}_{ij} (\epsilon) $ and $w^{(n-1)}_{jk} (\epsilon)$
\begin{displaymath}
	( w^{(n-1)}_{jk} (\epsilon) )^\dagger ( w^{(n)}_{ij}(\epsilon) )=\frac{ (w^{(n-1)}_{jk})^\dagger \cdot  w^{(n-1)}_{ik} - \beta^{(n-1)}_{ji} (w^{(n-1)}_{jk})^\dagger \cdot  w^{(n-1)}_{jk}}{|w^{(n-1)}_{ik} - \beta^{(n-1)}_{ji} w^{(n-1)}_{jk}|} \ ,
\end{displaymath}
which vanishes identically.

The remaining scalar products are proven to vanish by induction: using the definition of $w^{(m+1)}_{ij} (\epsilon )$, it is easy to show similarly that $w^{(n)}_{ij} (\epsilon )$ and $w^{(m+1)}_{ij} (\epsilon )$ are orthogonal provided $w^{(n)}_{ij} (\epsilon )$ and $w^{(m)}_{ij} (\epsilon )$ are orthogonal, proving the remaining orthogonality relations.

\end{proof}

\begin{corollary}
\emph{Span}$\{\lim_{ \epsilon\rightarrow 0} v_i,\lim_{ \epsilon\rightarrow 0} w^{(1)}_{ij} , \dots , \lim_{ \epsilon\rightarrow 0} w^{(N-1)}_{ij} \}$ has dimension $N$ and is isomorphic to the vector space spanned by the set of generalised eigenvectors of $M(0)$.
\end{corollary}
\begin{proof}
As the limits $\lim_{ \epsilon\rightarrow 0} w^{(n-1)}_{ij}$ are independent of the vectors we used to construct them, we can take one representative for each $w^{(n)}_{ij}$. From the previous proposition, the $N$ vectors that generate this space are all non-zero and orthogonal. This proves the first part of the corollary. 

By construction, the vector $\lim_{ \epsilon\rightarrow 0} w^{(n-1)}_{ij}$ is linear combination of  generalised eigenvectors up of rank $n$. As we have only one true eigenvector, $\text{Ker} \{ (M(0)-\lambda \mathbb{I})^n \} / \text{Ker} \{ (M(0)-\lambda \mathbb{I})^{n-1} \}$ is one-dimensional. Thus, we can stablish as isomorphism between $\text{Ker} \{ (M(0)-\lambda \mathbb{I})^n \}$ and Span$\{\lim_{ \epsilon\rightarrow 0} v_i,\lim_{ \epsilon\rightarrow 0} w^{(1)}_{ij} , \dots ,\lim_{ \epsilon\rightarrow 0} w^{(n-1)}_{ij}\}$, proving the second part of the corollary.
\end{proof}

Thus, we can claim that our method for constructing generalised eigenvectors is complete, meaning that it gives us all the generalised eigenvectors of $M(0)$, and indicates us what their rank is.

Despite considering here the case of a single Jordan block, the results of this section hold without any modification for any matrix with several Jordan blocks as long as all the eigenvalues have geometric multiplicity one, i.e.,  the eigenvalues associated with each Jordan Block are pairwise distinct. In addition, it is easy to check that the dimension of each Jordan block is equal to the number of eigenvectors that collapse to the same vector. 

\subsection{How to compute the limit} \label{howtolimit}

At this point, we have laid down an algorithm to find all the generalised eigenvectors of a defective matrix by perturbing it enough to become diagonalisable and computing the limit of vanishing perturbation of linear combinations of eigenvectors. Although the expressions we have provided for these linear combinations of eigenvectors were good for proving the completeness of the method, they are not so good for actual computations. In particular, they require us to compute scalar products of eigenvectors and norms of linear combinations of eigenvectors, which usually are far from straightforward in the context of the Bethe ansatz. Here we will provide some alternative expressions for the vectors $w_{ij}^{(n)}$ that are simpler for computational purposes.

The first simplification we can do concerns the norm in the denominator of (\ref{recipe}). The main purpose of this norm is to provide the correct power of $\epsilon$ for the vector $w_{ij}^{(n)}$ to be non-vanishing in the $\epsilon \rightarrow 0$ limit, as we are free to normalise our vectors as it suits us best. Thus, instead of considering the full norm, we are allowed to substitute it by the appropriate power of $\epsilon$, which is much easier to compute.

Regarding the numerator of (\ref{recipe}), we should recall that we required the condition $\lim_{\epsilon\rightarrow 0} (\lambda_2 - \lambda_1) (\alpha_1 \, v_1 + \alpha_2 \, v_2)=0$ for the limit $\lim_{\epsilon\rightarrow 0} (\alpha_1 \, v_1 + \alpha_2 \, v_2)$ to be well-defined. Similar conditions hold for all $w_{ij}^{(n)}$ vectors. Thus, we can choose the numerator to be a linear combination of eigenvectors with finite coefficients (with the restriction that these coefficients do not vanish when $\epsilon\rightarrow 0$) that vanishes in the $\epsilon \rightarrow 0$ limit. Nevertheless, if we want to make our computations even simpler, we might want to choose this linear combination to be orthogonal to the generalised eigenvector computed in the previous step. Notice that, if this is done at every step, this linear combination will be orthogonal to all the generalised eigenvectors computed in previous steps. We should stress again that what we get after doing these simplifications is not actually a generalised eigenvector, but a generalised eigenvector up to a linear combination of the vectors we obtained in previous steps. This is not a bad trade-off, as the algorithm computes solutions to the eigenvector problem associated with the matrix $(M-\lambda_i \mathbb{I})^n $, so it is entirely possible that the results we were obtaining were already generalised eigenvector up to a linear combination of the vectors we obtained in previous steps.

Finally, we want to stress that we have to compute the limit of a linear combination of eigenvectors of the perturbed diagonal matrix. If we compute the limit of the difference between an eigenvector of the diagonalisable matrix and a linear combination of generalised eigenvectors of the defective matrix, the algorithm might give us the wrong generalised eigenvector. This is because, in general, 
\begin{equation}
	\lim_{\epsilon \rightarrow 0} \frac{ w^{(n-1)}_{ji} - \beta^{(n-1)}_{kj} w^{(n-1)}_{ki}}{|w^{(n-1)}_{ji} - \beta^{(n-1)}_{kj} w^{(n-1)}_{ki}|} \neq \lim_{\epsilon \rightarrow 0} \frac{ w^{(n-1)}_{ji} - \lim_{\epsilon \rightarrow 0} [\beta^{(n-1)}_{kj} w^{(n-1)}_{ki}]}{|w^{(n-1)}_{ji} - \beta^{(n-1)}_{kj} w^{(n-1)}_{ki}|} \ .
\end{equation}
We can see the consequences of this inequality very clearly in the following example: Consider the matrix
\begin{displaymath}
	M(\epsilon) =\begin{pmatrix}
		0 & 1 & \epsilon^2 \\
		\epsilon^2 & 0 & 1 \\
		0 & 0 & 0
	\end{pmatrix} \ .
\end{displaymath}
This matrix has the following eigenvalues and eigenvectors
\begin{align*}
	\lambda_1 &= 0 \ , & v_1&= (1,\epsilon^4 , -\epsilon^2 ) \ , \\
	\lambda_2 &= -\epsilon \ , & v_2&= (1,-\epsilon , 0 ) \ , \\
	\lambda_3 &= \epsilon \ , & v_3&= (1,\epsilon , 0 ) \ .
\end{align*}
We can check that $\lim_{\epsilon\rightarrow 0} v_i=(1,0 , 0 )$ for the three vectors. If we move now to $w^{(1)}$ we find that
\begin{align*}
	&\frac{v_1 - \beta_{21} v_2}{|v_1 - \beta_{21} v_2|} \approx ( \epsilon^4 , 1, -\epsilon) \ , & &\frac{v_1 - \lim_{\epsilon \rightarrow 0} (\beta_{21} v_2)}{|v_1 - \beta_{21} v_2|} \approx ( 0 , -\epsilon^2, 1) \ ,\\
	&\frac{v_2 - \beta_{32} v_3}{|v_2 - \beta_{32} v_2|} \approx ( -\epsilon , 1, 0) \ , & &\frac{v_2 - \lim_{\epsilon \rightarrow 0} (\beta_{32} v_3)}{|v_2 - \beta_{32} v_3|} \approx ( 0 , 1, 0) \ .
\end{align*}
As we can see, if we consider the difference between an eigenvector and the limit of an eigenvector, we will wrongly assume that there are two generalised eigenvectors of rank 2 instead of just one, which contradicts theorem \ref{dimker} because we have only one eigenvector.

Although the authors of \cite{Gainutdinov:2016pxy} use the difference of the vector and the limit of a second vector to compute generalised eigenvectors, there is no issue with it in their context. The reason being that all their Jordan blocks are of size two, so there are not enough generalised eigenvectors to create the ``misplacement'' we have seen in this example.

\section{Exceptional point with degenerate Jordan blocks}

The algorithm from the previous sections works perfectly as long as the matrix we are considering is composed of Jordan blocks associated with different eigenvalues at the exceptional point. This is thanks to having exactly one generalised eigenvector of a given rank for each eigenvalue. However, the generalisation to the case where we have two or more Jordan cells associated with the same eigenvalue can give rise to issues because we lack this uniqueness.

As in the previous section, we will first examine a particular example that captures most of the peculiarities of the problem, and then discuss how to generalise the results from the previous section in depth.

\subsection{A specific example: five-sites chain with three different excitations and zero total momentum}\label{5sites}

In this section, we take a look at the simplest case where the strongly twisted limit of the Hamiltonian (\ref{eq:twistedXXX}) develops more than one Jordan block for the case where the total momentum of the excitations vanishes (that is, eigenstates of the operator $U$ with eigenvalue $1$). This corresponds to a spin chain of length $5$ with two excitations of one flavour and one of the other ($L=5$, $M=3$, $K=1$). According to \cite{StaudacherAhn}, the Jordan normal form of the Hamiltonian has a Jordan cell of size 5 and a Jordan cell of size 1. Although it is outside of the range of validity of the procedure we explained before, as the two Jordan blocks are indistinguishable, we will still apply it to check if we can retrieve any useful information about the Jordan structure of the Hamiltonian.

For the case of vanishing total momentum, the Hamiltonian can be written in the following form
\begin{equation}
	\mathbf{\tilde{H}}_{(q_1,q_2,q_3)}=\begin{pmatrix}
	2 & q_3 & q_1 & q_2 & 0 & 0\\
	\frac{1}{q_3} & 0 & q_3 & q_3 & q_1+q_2 & 0 \\
	\frac{1}{q_1} & \frac{1}{q_3} & 1 & 0 & q_3 & q_1 \\
	\frac{1}{q_2} & \frac{1}{q_3} & 0 & 1 & q_3 & q_2 \\
	0 & \frac{1}{q_1}+\frac{1}{q_2} & \frac{1}{q_3} & \frac{1}{q_3} & 0 & q_3 \\
	0 & 0 & \frac{1}{q_1} & \frac{1}{q_2} & \frac{1}{q_3} & 2
	\end{pmatrix}\ .
\end{equation}
Sadly, this Hamiltonian is very difficult to diagonalise. Consequently, we decided to set some of its strictly lower diagonal elements to zero to simplify this example, as it makes the expressions for the eigenvalues and eigenvectors much simpler. In principle, we are allowed to do that because setting them to zero by hand modifies neither the position of the exceptional point nor the matrix at the exceptional point. In particular, we will focus on the following two matrices
\begin{equation}
 \tilde{H}^{(1)}=\begin{pmatrix}
	2 & q_3 & q_1 & q_2 & 0 & 0\\
	0 & 0 & q_3 & q_3 & q_1+q_2 & 0 \\
	0 & 0 & 1 & 0 & q_3 & q_1 \\
	0 & 0 & 0 & 1 & q_3 & q_2 \\
	0 & 0 & 0 & 0 & 0 & q_3 \\
	0 & 0 & 0 & 0 & \frac{1}{q_3} & 2
	\end{pmatrix} \ , \qquad  \tilde{H}^{(2)}=\begin{pmatrix}
	2 & q_3 & q_1 & q_2 & 0 & 0\\
	0 & 0 & q_3 & q_3 & q_1+q_2 & 0 \\
	0 & 0 & 1 & 0 & q_3 & q_1 \\
	0 & \frac{1}{q_3} & 0 & 1 & q_3 & q_2 \\
	0 & 0 & 0 & 0 & 0 & q_3 \\
	0 & 0 & 0 & 0 & \frac{1}{q_3} & 2
	\end{pmatrix} \ ,
\end{equation}
but similar results hold for other choices.

To simplify our discussion, we will denote by $\hat{u}_i$ the vector (with appropriate dimension) where component $i$ is equal to $1$ and all the remaining components are equal to zero. 

After substituting $q_i\rightarrow \frac{\xi_i}{\epsilon}$, the six eigenvectors of $\tilde{H}^{(1)}$ take the form
\begin{align*}
	v_1 =& \left( 1 , 0, 0, 0, 0, 0 \right) \ , \\
	v_2 =& \left( \xi_3 , -2\epsilon, 0, 0, 0, 0 \right) \ , \\
	v_3 =& \left( \xi_3^2 + \xi_2 \epsilon , -2 \xi_3 \epsilon, 0, -\epsilon^2, 0, 0 \right) \ , \\
	v_4 =& \left( \xi_3^2 + \xi_1 \epsilon , -2 \xi_3 \epsilon, 0, -\epsilon^2, 0, 0 \right) \ , \\
	v_5 =& \left( \xi_3^4-\frac{1+2\sqrt{2}}{2+\sqrt{2}} (\xi_1+\xi_2)\xi_3^2 \epsilon + \frac{2-\sqrt{2}}{2+\sqrt{2}} (\xi_1^2+\xi_2^2) \frac{\epsilon^2}{2} , -\frac{4+3\sqrt{2}}{2+\sqrt{2}} \xi_3^3 \epsilon + \frac{8+5\sqrt{2}}{2+\sqrt{2}} (\xi_1+\xi_2)\xi_3 \frac{\epsilon^2}{2} , \right. \\
	&\left. \frac{\xi_3^2 \epsilon^2}{2}+\frac{1-\sqrt{2}}{2} \xi_1 \epsilon^3 ,\frac{\xi_3^2 \epsilon^2}{2}+\frac{1-\sqrt{2}}{2} \xi_2  \epsilon^3 , -\frac{1+\sqrt{2}}{2+\sqrt{2}} \xi_3 \epsilon^3 , \frac{\epsilon^4}{2+\sqrt{2}} \right) \ , \\
	v_6 =& \left( \xi_3^4-\frac{1-2\sqrt{2}}{2-\sqrt{2}} (\xi_1+\xi_2)\xi_3^2 \epsilon + \frac{2+\sqrt{2}}{2-\sqrt{2}} (\xi_1^2+\xi_2^2) \frac{\epsilon^2}{2} , -\frac{4-3\sqrt{2}}{2-\sqrt{2}} \xi_3^3 \epsilon + \frac{8-5\sqrt{2}}{2-\sqrt{2}} (\xi_1+\xi_2)\xi_3 \frac{\epsilon^2}{2} , \right. \\
	&\left. \frac{\xi_3^2 \epsilon^2}{2}+\frac{1+\sqrt{2}}{2} \xi_1 \epsilon^3 ,\frac{\xi_3^2 \epsilon^2}{2}+\frac{1+\sqrt{2}}{2} \xi_2  \epsilon^3 , -\frac{1-\sqrt{2}}{2-\sqrt{2}} \xi_3 \epsilon^3 , \frac{\epsilon^4}{2-\sqrt{2}} \right) \ .
\end{align*}
We can check that all these six vectors become $\hat{u}_1=(1,0,0,0,0,0)$ when we compute the limit $\epsilon\rightarrow 0$. In addition, if we construct the following five vectors
\begin{equation}
	w_{i,1}^{(1)}=\frac{v_i - (v_i\cdot v_1) v_1}{\epsilon} \ ,
\end{equation}
we get that they all become proportional to $\hat{u}_2=(0,1,0,0,0,0)$ in the limit of vanishing $\epsilon$. If we consider now the following four vectors
\begin{equation}
	w_{i,2}^{(2)}=\frac{w_{i,1}^{(1)} - (w_{2,1}^{(1)}\cdot w_{i,1}^{(1)}) w_{2,1}^{(1)}}{\epsilon} \ ,
\end{equation}
we find that they do not approach the same vector when $\epsilon$ approaches $0$. Instead, we find that they are linear combinations of the vectors $\hat{u}_3$ and $\hat{u}_4$. However, theorem~\ref{dimker} tells us that it is impossible to obtain two generalised eigenvectors at this step because we naïvely have just one true eigenvector. We can check that indeed both vectors are generalised eigenvectors of rank 2 (up to a linear combination with $\hat{u}_2$) with respect to our true eigenvector. This is only possible if they differ by a true eigenvector. After some algebra, we find that the vector $(0,\xi_2-\xi_1,\xi_3,-\xi_3,0,0)$ is indeed another eigenvector of the strongly twisted Hamiltonian. If we continue the process, we find
\begin{align}
	\lim_{\epsilon \rightarrow 0} w^{(3)} &=(0,0,0,0,1,0)=\hat{u}_5 \ , & \lim_{\epsilon \rightarrow 0} w^{(4)} &=(0,0,0,0,0,1)=\hat{u}_6 \ .
\end{align}

The structure of the limits can be summarised by the following diagram:
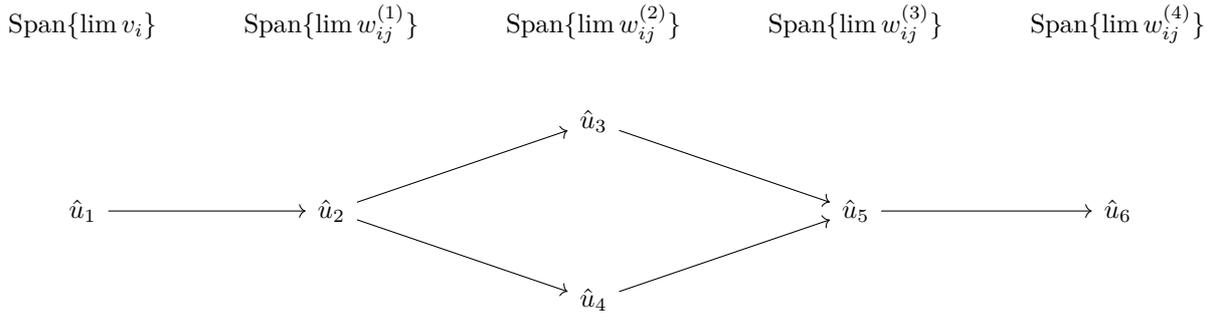
\begin{center}
\begin{tikzcd}
\text{Span} \{\lim v_i\} & \text{Span} \{\lim w^{(1)}_{ij} \} & \text{Span} \{\lim w^{(2)}_{ij} \} & \text{Span} \{\lim w^{(3)}_{ij} \} & \text{Span} \{\lim w^{(4)}_{ij} \} \\
 & & \hat{u}_3 \arrow[rd] & & \\
\hat{u}_1 \arrow[r] & \hat{u}_2 \arrow[ru] \arrow[rd] & & \hat{u}_5 \arrow[r] & \hat{u}_6 \\
 & & \hat{u}_4 \arrow[ru] & & 
\end{tikzcd}
\end{center}
At the left most part we have placed the true eigenvector and, as we move one step to the right, we get the set of vectors obtained after applying once our ``limit of linear combination of eigenvectors that coalesce'' recipe. Although we might think that this is an indication of the existence of a Jordan block of size $5$ and a Jordan block of size $1$, we should not be so eager to jump to conclusions. In order to check that, we can apply the definition of generalised eigenvector (\ref{geneigenvdefinition}). After some tedious algebra, one finds this identification to be correct
\begin{align*}
	&\mathbf{\hat{H}}^5 \frac{2\xi_3^3 \hat{u}_6-3 \xi_3^2 (\xi_1 + \xi_2) \hat{u}_5 +2\xi_3 (\xi_1^2 + \xi_2^2 + 3\xi_1 \xi_2)(\hat{u}_3 +\hat{u}_4) -2(\xi_1^3 + \xi_2^3 +4\xi_1^2 \xi_2 + 4\xi_1 \xi_2^2) \hat{u}_2}{4 \xi_3^7} = \\
	&\mathbf{\hat{H}}^4 \frac{2\xi_3^2 \hat{u}_5 -\xi_3 (3\xi_1 + \xi_2) \hat{u}_4 -\xi_3 (\xi_1 + 3\xi_2) \hat{u}_3+ (\xi_1^2 + \xi_2^2 + 6\xi_1 \xi_2) \hat{u}_2}{4 \xi_3^5}=\mathbf{\hat{H}}^3 \frac{\xi_3 (\hat{u}_3 +\hat{u}_4)-(\xi_1 + \xi_2) \hat{u}_2}{2\xi_3^3}= \\
	&\mathbf{\hat{H}}^2 \frac{\hat{u}_2}{\xi_3}= \mathbf{\hat{H}}\hat{u}_1=0 \ , \\
	&\mathbf{\hat{H}} [\xi_3 (\hat{u}_3 -\hat{u}_4) -(\xi_1 - \xi_2) \hat{u}_2 ] =0 \ ,
\end{align*}
where $\mathbf{\hat{H}}=\lim_{\epsilon \rightarrow 0} \mathbf{\tilde{H}}_{(\frac{\xi_1}{\epsilon},\frac{\xi_2}{\epsilon},\frac{\xi_3}{\epsilon})}$. This result is in agreement with the one presented in \cite{StaudacherAhn}. We will later show a better method to compute the generalised eigenvectors once we know the eigenvector at the bottom of the Jordan chain.

Let us now examine the other matrix. The six eigenvectors of $\tilde{H}^{(2)}$ take the form
\begin{align*}
	v_1 =& \left( 1 , 0, 0, 0, 0, 0 \right) \ , \\
	v_2 =& \left( \frac{(1+\sqrt{5}) \xi_3^2 -2\xi_2 \epsilon}{3+\sqrt{5}} , -(1+\sqrt{5}) \xi_3 \frac{\epsilon}{2} , 0, \epsilon^2 , 0, 0 \right) \ , \\
	v_3 =& \left( \frac{(1-\sqrt{5}) \xi_3^2 -2\xi_2 \epsilon}{3-\sqrt{5}} , -(1-\sqrt{5}) \xi_3 \frac{\epsilon}{2} , 0, \epsilon^2 , 0, 0 \right) \ , \\
	v_4 =& \left( \xi_1-\xi_2 , 0, -\epsilon , \epsilon , 0, 0 \right) \ , \\
	v_5 =& \left( 2(2-\sqrt{2}) \xi_3^4 -[(1-3\sqrt{2})\xi_1 +(3-5\sqrt{2}) \xi_2]\xi_3^2 \epsilon + [(1+\sqrt{2}) \xi_1^2 + (3-\sqrt{2}) \xi_1 \xi_2 + 4\xi_2^2]\epsilon^2 , \right. \\
	&  -\frac{4 (1- \sqrt{2}) \xi_3^3 \epsilon - 2 (3-\sqrt{2}) (\xi_1 + \xi_2) \xi_3 \epsilon^2}{2+\sqrt{2}} , -(1-\sqrt{2})\xi_3^2 \epsilon^2 +\xi_1 \epsilon^3, \\
	&\left. (1-\sqrt{2}) [(1-2\sqrt{2}) \xi_3^2 \epsilon^2 -4\xi_2 \epsilon^3 -(3-\sqrt{2}) \xi_1 \epsilon^3 , (2-\sqrt{2}) \xi_3 \epsilon^3 , \sqrt{2} \epsilon^4 \right) \ , \\
	v_6 =& \left( -2(2+\sqrt{2}) \xi_3^4 +[(1+3\sqrt{2})\xi_1 +(3+5\sqrt{2}) \xi_2]\xi_3^2 \epsilon - [(1-\sqrt{2}) \xi_1^2 + (3+\sqrt{2}) \xi_1 \xi_2 + 4\xi_2^2]\epsilon^2 , \right. \\
	&  \frac{4 (1+ \sqrt{2}) \xi_3^3 \epsilon - 2 (3+\sqrt{2}) (\xi_1 + \xi_2) \xi_3 \epsilon^2}{2+\sqrt{2}} , (1+\sqrt{2})\xi_3^2 \epsilon^2 -\xi_1 \epsilon^3, \\
	&\left. -(1+\sqrt{2}) [(1+2\sqrt{2}) \xi_3^2 \epsilon^2 -4\xi_2 \epsilon^3 -(3+\sqrt{2}) \xi_1 \epsilon^3 , -(2+\sqrt{2}) \xi_3 \epsilon^3 , \sqrt{2} \epsilon^4 \right) \ .
\end{align*}
We can check that all these six vectors become $\hat{u}_1$ when we compute the limit $\epsilon\rightarrow 0$. However, if we construct the following five vectors
\begin{equation}
	w_{i,1}^{(1)}=\frac{v_i - (v_i\cdot v_1) v_1}{\epsilon} \ ,
\end{equation}
we get that they become proportional to a linear combination of the vectors $\hat{u}_2$ and $\hat{u}_3-\hat{u}_4$. We can check that indeed both vectors are generalised eigenvectors of rank 2 with respect to our true eigenvector. If we continue the process, we find
\begin{align}
	\lim_{\epsilon \rightarrow 0} w^{(2)} &=\hat{u}_3+\hat{u}_4 \ , & \lim_{\epsilon \rightarrow 0} w^{(3)} &=\hat{u}_5 \ , & \lim_{\epsilon \rightarrow 0} w^{(4)} &=\hat{u}_6 \ .
\end{align}
The results can be summarised in the following diagram:
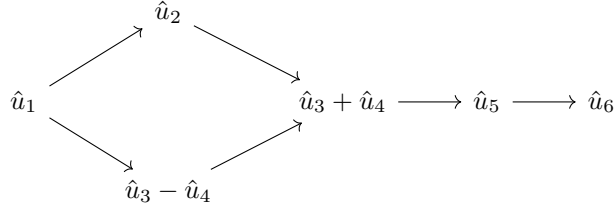
\begin{center}
\begin{tikzcd}
 & \hat{u}_2 \arrow[rd] \\
\hat{u}_1  \arrow[ru] \arrow[rd] & & \hat{u}_3 +\hat{u}_4 \arrow[r] & \hat{u}_5 \arrow[r] & \hat{u}_6 \\
 & \hat{u}_3-\hat{u}_4 \arrow[ru]
\end{tikzcd}
\end{center}
The Hamiltonian $\tilde{H}^{(2)}$ gives us a very similar structure to the one we found for the  Hamiltonian $\tilde{H}^{(1)}$, with the exception that the ``fake generalised eigenvector of rank 3'' now appears as a ``fake generalised eigenvector of rank 2''.

In order to see that this is not a special feature of the case we are considering, we can repeat the same construction for the case of a spin chain of length $7$ with two excitations of one flavour and one of the other ($L=7$, $M=3$, $K=1$). We will not reproduce the whole process here, but the final result can be summarised in the following diagram
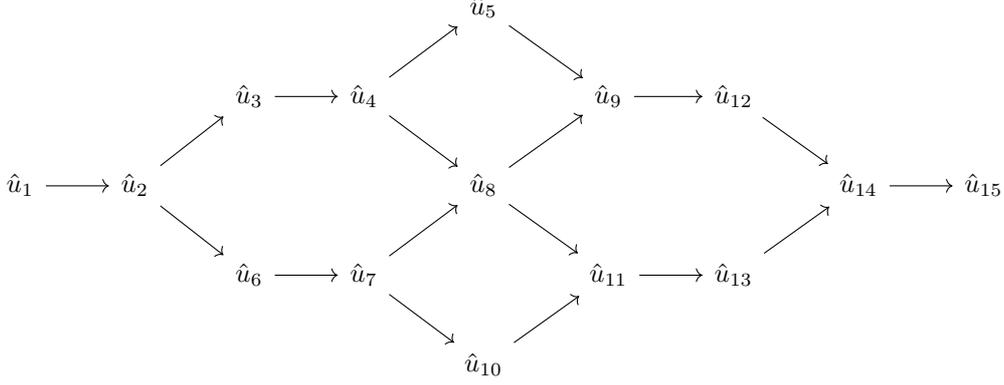
\begin{center}
\begin{tikzcd}
   & & & & \hat{u}_5 \arrow[rd]  \\
   & & \hat{u}_3 \arrow[r] & \hat{u}_4 \arrow[ru] \arrow[rd] & & \hat{u}_9 \arrow[r] & \hat{u}_{12}  \arrow[rd] \\
\hat{u}_1 \arrow[r] & \hat{u}_2 \arrow[ru] \arrow[rd] & & & \hat{u}_8 \arrow[ru] \arrow[rd] & & & \hat{u}_{14} \arrow[r] & \hat{u}_{15}  \\
   & & \hat{u}_6 \arrow[r] & \hat{u}_7 \arrow[ru] \arrow[rd] & & \hat{u}_{11} \arrow[r] & \hat{u}_{13} \arrow[ru]  \\
   & & & & \hat{u}_{10} \arrow[ru] 
\end{tikzcd}
\end{center}
From this diagram we can postulate that this case has three Jordan blocks of size $9$, $5$ and $1$ respectively. This result again matches exactly the one collected in Table 1 in \cite{StaudacherAhn} and the computation from \cite{Ahn:2021emp}.

\subsection{The defective limit of a diagonalisable matrix: Degenerate Jordan blocks} \label{nondistinguishableJB}

In the previous examples we have seen that, although we are outside the range of applicability of the method we proposed in section~\ref{distinguishableJB}, it seems to capture enough of the structure of the Jordan blocks. In this section, we will explore which parts of the method remain unaltered and which parts change when several Jordan blocks share the same eigenvalues.

Let us consider again a matrix $M(\epsilon)$ that is diagonalisable for generic values of $\epsilon$ but that becomes defective for $\epsilon=0$. Here, we will assume that two or more Jordan cells have the same eigenvalue. First, we should revisit equation~(\ref{important})
\begin{equation}
	\left[M(0)- \mathbb{I}\lim_{\epsilon\rightarrow 0} \lambda_i\right]  \left[ \lim_{\epsilon\rightarrow 0}v_i\right]=0 \ .
\end{equation}
As we commented before, this equation is giving us a sufficient condition on the limit of the eigenvectors, not a necessary one. This means that the limit of any eigenvector of $M(\epsilon)$ will be an eigenvector of $M(0)$ but it does not guarantee that all the eigenvectors of $M(0)$ can be obtained in this way. In the case of distinguishable Jordan blocks, we did not have this issue because both spaces were one-dimensional, so we were able to unequivocally identify one with the other. However, if we have several Jordan blocks associated with the same eigenvalue, the dimension of Ker$\{(M(0)-\lambda \mathbb{I})\}$ is larger than one while the dimension of Span$\{\lim_{\epsilon\rightarrow 0}v_i\}$ depends on the form of $M(\epsilon)$. In fact, these two dimensions do not coincide in any of the examples at the beginning of this section, with true eigenvectors appearing as fake generalised eigenvector of rank 2, 3 or 5. As Jordan chains become entangled, we will denote this effect as \emph{chain mixing}.

The core reason behind this chain mixing is that there exist linear combinations of eigenvectors of $M(\epsilon)$ with coefficients that diverge at $\epsilon =0$ that fulfil
\begin{equation}
\left[M(\epsilon)- \mathbb{I} \lambda_1 (\epsilon) \right] \sum_{i=1}^{n-1} \alpha_i v_i (\epsilon) = \mathcal{O} (\epsilon) \qquad \text{while} \qquad \lim_{\epsilon\rightarrow 0} \sum_{i=1}^{n-1} \alpha_i v_i (\epsilon)= \mathcal{O} (\epsilon^0) \ .
\end{equation}
That is, they are true eigenvectors of $M(0)$ but they cannot be extended to true eigenvectors of $M(\epsilon)$. Obviously, which of the true eigenvectors of $M(0)$ cannot be extended to a true eigenvector of $M(\epsilon)$ depends only on how $M(\epsilon)$ approaches $M(0)$.  As an example of this effect, consider the four following matrices, all of which have the same limit at $\epsilon=0$
	\begin{align}
	&\begin{pmatrix}
	0 & 1 & \epsilon^2 \\
	\epsilon^2 & \epsilon^4 & \epsilon^2 \\
	0 & 0 & \epsilon^6
	\end{pmatrix} & &\begin{pmatrix}
	0 & 1 & \epsilon^2 \\
	\epsilon^2 & \epsilon^4 & \epsilon \\
	0 & 0 & \epsilon^6
	\end{pmatrix} & &\begin{pmatrix}
	0 & 1 & \epsilon^2 \\
	\epsilon^5 & 0 & 0 \\
	\epsilon & 0 & 0
	\end{pmatrix} & \begin{pmatrix}
	0 & 1 & \epsilon^2 \\
	0 & \epsilon^4 & \epsilon^2 \\
	0 & 0 & \epsilon^6
	\end{pmatrix} \ . \label{EX1}
\end{align}
Although they approach the same matrix as $\epsilon$ approaches zero, their eigenvectors have completely different analytic structure. The results of applying our recipe to each of them can be summarised in the following four diagrams\footnote{As we indicated in section~\ref{howtolimit}, we have to be careful when computing the limit of the vectors $w^{(n)}$, particularly in the case of the third matrix.}
	\begin{align*}
	&\begin{tikzcd}[ampersand replacement=\&]
   \hat{u}_1 \arrow[r] \& \hat{u}_2  \\
   \hat{u}_3 \arrow[ru]
\end{tikzcd} & &\begin{tikzcd}[ampersand replacement=\&]
   \hat{u}_1 \arrow[r] \arrow[rd] \& \hat{u}_2 \\
   \& \hat{u}_3
\end{tikzcd} & &\begin{tikzcd}[ampersand replacement=\&]
   \hat{u}_3 \arrow[r] \arrow[r] \& \hat{u}_1 \arrow[r]  \& \hat{u}_2
\end{tikzcd} & &\begin{tikzcd}[ampersand replacement=\&]
   \hat{u}_1 \arrow[r] \arrow[r] \& \hat{u}_2 \arrow[r]  \& \hat{u}_3
\end{tikzcd}
\end{align*}
From these four examples, we see that having several Jordan cells associated with the eigenvalue muddies the analysis, making it seemingly impossible to separate the two Jordan chains of $M(0)$. This justifies the term ``chain mixing'' we decided to use. The problem only gets worse if we consider bigger Jordan blocks. Consider the following example

	\begin{align}
	\begin{pmatrix}
	\epsilon & 1 & \epsilon & 0 & 0 \\
	0 & \epsilon^4 & 1 & \epsilon & 0 \\
	0 & 0 & \epsilon^7 & \epsilon^9 & 0 \\
	0 & 0 & 0 & \epsilon^6 & 1 \\
	0 & 0 & 0 & 0 & \epsilon^8
	\end{pmatrix} \ . \label{EX2}
	\end{align}
the diagram representing the structure of the vectors we obtain takes the form

	\hfill
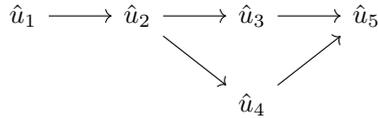
\begin{tikzcd}
	   \hat{u}_1 \arrow[r] & \hat{u}_2 \arrow[r] \arrow[rd] & \hat{u}_3 \arrow[r] & \hat{u}_5 \\
	    & & \hat{u}_4 \arrow[ru]
	\end{tikzcd}\hfill \null
	
	\noindent Thus, at first sight we cannot distinguish if $\lim_{\epsilon\rightarrow 0} w^{(3)}=\hat{u}_5$ is a generalised eigenvector of rank $4$, a generalised eigenvector of rank $2$ or even a true eigenvector.

Although the situation looks a bit hopeless, there is more information in these examples that what it seems at first sight. In fact, we want to make two observations: the method to compute generalised eigenvectors is still complete (meaning that we find all the generalised eigenvectors) and Jordan chains are not broken (meaning that if we find a generalised eigenvector of rank $n$ at a given step, its associated generalised eigenvector of rank $n-1$ has to appear in the previous step). The remaining of this section is devoted to analysing these two observations and their consequences.

First of all, in spite of being outside its range of applicability, the method to compute the generalised eigenvalues we proposed in section~\ref{distinguishableJB} is complete also in this case, meaning that it gives us the full set of generalised eigenvectors of $M(0)$. The proof of that is exactly the same in the situation of distinguishable Jordan blocks: the method give us non-zero and non-divergent orthogonal vectors, and we can apply it as many times as eigenvectors of $M(\epsilon)$ we have. In addition, it is still true that all the vectors we obtain by naïvely computing the limit $\epsilon\rightarrow 0$ of the eigenvectors of $M(\epsilon)$ are true eigenvectors of $M(0)$.

Let us now analyse the second observation. If a ``misplaced'' eigenvector appears at the $n$-th step of our computation, we may find a generalised eigenvector of rank $m$ at the $n+m-1$-th step. The proof of this is relatively similar to the proof for the case of distinguishable Jordan block. Given the fact that there exists a linear combination of vectors fulfilling
\begin{displaymath}
	\left[M(\epsilon)- \mathbb{I} \lambda_n (\epsilon) \right] \sum_{i=1}^{n} \alpha_i v_i = \mathcal{O} (\epsilon) \ ,
\end{displaymath}
it is easy to see that a linear combination of $n+m-1$ vectors will fulfil
\begin{equation}
	\prod_{j=1}^m\left[M(\epsilon)- \mathbb{I} \lambda_{n+j-1} (\epsilon) \right] \sum_{i=1}^{n+m-1} \beta_i v_i =\mathcal{O} (\epsilon) \ ,
\end{equation}
as long as we choose $\beta_i=\gamma \frac{\alpha_i}{\prod_{j=1}^{m} (\lambda_i-\lambda_{n+j-1})}$ for $1\leq i<n$, where $\lambda_j$ is the eigenvalue associated with $v_j$ and $\gamma$ is a constant. This means that such linear combinations contain information about generalised eigenvectors of higher rank, which can be extracted by applying the same procedure we used for the case of distinguishable Jordan blocks.

With this information, we can take another look at the issue of vectors that seemingly contradict theorem~\ref{dimker}. If we where in the case of distinguishable Jordan blocks, the number of vectors we have obtained at the $n$-th step would have been the dimension of the space Ker$\{(M(0)-\lambda_i \mathbb{I})^n\}/$Ker$\{(M(0)-\lambda_i \mathbb{I})^{n-1}\}$. This would have implied, via said theorem, that we would have obtained the same number or less vectors at the $n$-th step than at the first step. Thus, if the dimension of the vector space we obtain at a given step is larger than the dimension of the vector space we obtained at the previous step, that is an unambiguous sign that we have found an eigenvector in disguise. In fact, we can say that if we consider the dimension of each of the vector spaces we get at each step of the process, the largest of those dimensions will be a lower bound on the number of Jordan blocks associated with the eigenvalue we are considering.

We want to end the section by showing a reliable but cumbersome procedure to separate the mixed Jordan chains. We have discussed above that, despite the chain mixing issue, our method is able to identify some true eigenvectors of $M(0)$ beyond the shadow of a doubt. Thus, if we are able to completely reconstruct a Jordan chain from its generalised eigenvector of rank 1, we can separate the mixed chains. This reconstruction can be done if we combine the corollary~\ref{lowesttohighest} and the definition of generalised eigenvectors (\ref{geneigenvdefinition}). The procedure goes as follows. Using the corollary~\ref{lowesttohighest}, we can identify the right generalised eigenvector of rank 1 of $M(0)$ with a left generalised eigenvector of highest rank (which is equivalent to a right generalised eigenvector of highest rank of $M^\dagger (0)$). Now, we can use the definition of generalised eigenvectors (\ref{geneigenvdefinition}) to find the left generalised eigenvector of rank 1 associated with it. Using again the corollary~\ref{lowesttohighest}, we can identify this vector with the right generalised eigenvector of highest rank of the same Jordan chain than the generalised eigenvector of rank 1 we started with. At this point, we only have to descend the Jordan chain using the definition (\ref{geneigenvdefinition}). To summarise it, if $v$ is a generalised eigenvector of rank 1 of $M(0)$, then the generalised eigenvectors of its associated Jordan chain can be constructed as follows
\begin{equation}
	M(0)^l \left[ M (0)^\dagger \right]^m v \ ,
\end{equation}
where $0\leq l \leq m$, and $m$ is an non-negative integer such that $\left[ M (0)^\dagger \right]^m v \neq 0$ and $\left[ M (0)^\dagger \right]^{m+1} v =0$.

This issue of chain mixing was also pointed out in \cite{Ahn:2021emp} under the name of \emph{unexpected shortening}, although the authors do not provide a method to deal with it.

\section{Computing generalised eigenvectors for the eclectic spin chain}

In this section, we plan to find the generalised eigenvectors of the eclectic spin chain with a general number of excitations of type $2$ and only one excitation of type $3$, i.e. $K=1$ and $M > 1$. Recall that the $M=K=1$ case is diagonalisable and therefore not particularly interesting to us. Our first goal will be  to construct the eigenvectors of the twisted Hamiltonian (\ref{eq:twistedXXX}) for general values of the twist parameters $q_i$ by means of a modified version of the Nested Coordinate Bethe Ansatz (NCBA). Using these results, we will compute their behaviour at large values of the twist parameters $q_i$ and apply the algorithm presented in the previous sections to extract all the eigenvectors and generalised eigenvectors associated with the $K=1$ sector of the Hamiltonian.

\subsection{The Coordinate Bethe Ansatz for Twisted Spin Chains}

Our first step will be to construct the eigenvectors of the twisted Hamiltonian (\ref{eq:twistedXXX}) for the case of arbitrary length $L$ and excitations with two different flavours. The most straightforward method to achieve this is a slightly modified version of the NCBA. The original method is described in section II.O of \cite{Sutherland}, as well as in section 2 of \cite{Beisert:2005fw}.\footnote{Another version of the NCBA using Zamolodchikov-Faddeev operators can be found in \cite{MariusZF}.}

We should first clarify that the twisted Hamiltonian (\ref{eq:twistedXXX}) does not commute with the regular permutation operator. Nevertheless, it commutes with the shift operator $U$ that we defined in previous sections, meaning that we can apply Bloch's theorem despite the twists. Although this allows us to assume a plane-wave ansatz, we still have to include additional factors proportional to the twist parameters $q_i$. In particular, we will consider the following ansatz for the case of two excitations
\begin{align}\label{M2wavefunction}
	|\psi_{23} (p_1 , p_2) \rangle &=\sum_{1 \leq n_1 < n_2 \leq L} \left[ A_{23} e^{i(p_1  n_1 + p_2  n_2)} \frac{q_3^{n_1}}{q_2^{n_2}} + \tilde{A}_{23} e^{i(p_1  n_2 + p_2  n_1)} \frac{q_3^{n_1}}{q_2^{n_2}} \right] S^{2,+}_{n_1} S^{3,+}_{n_2} |0\rangle \notag \\
	&+\sum_{1 \leq n_1 < n_2 \leq L} \left[ A_{32} e^{i(p_1  n_1 + p_2  n_2)} \frac{q_3^{n_2}}{q_2^{n_1}} + \tilde{A}_{32} e^{i(p_1  n_2 + p_2  n_1)} \frac{q_3^{n_2}}{q_2^{n_1}} \right] S^{3,+}_{n_1} S^{2,+}_{n_2}|0\rangle \ ,
\end{align}
where we define the pseudo-vacuum as the tensor product of $L$ states of type 1, i.e. $|0\rangle=\otimes_{n=1}^L |1\rangle$, and $S_k^{n,+}$ is the operator that transforms a state of type $1$ in the $k$-th term of the tensor product into a state of type $n$. To simplify our notation, we will drop the dependence on the momenta and always understand momenta $p_1$ and $p_2$ unless otherwise is stated.

Substituting this ansatz into the Schrödinger equation $\mathbf{\tilde{H}}_{(q_1,q_2,q_3)}|\psi_{23} \rangle=E|\psi_{23} \rangle$ gives us the following relations
\begin{align}
\label{boi}
	E&=L-4+2\cos (p_1) + 2\cos(p_2 )  \ , \\
	\begin{pmatrix}
		\tilde{A}_{32} \\
		\tilde{A}_{23}
	\end{pmatrix} &= \begin{pmatrix}
		\frac{1}{q_1 q_2 q_3} \frac{e^{i p_2}-e^{i p_1}}{1-2 e^{i p_1} + e^{i (p_1+p_2)}} & \frac{-(1-e^{i p_1})(1-e^{i p_2})}{1-2 e^{i p_1} + e^{i (p_1+p_2)}} \\
		\frac{-(1-e^{i p_1})(1-e^{i p_2})}{1-2 e^{i p_1} + e^{i (p_1+p_2)}} &  q_1 q_2 q_3 \frac{e^{i p_2}-e^{i p_1}}{1-2 e^{i p_1} + e^{i (p_1+p_2)}}
	\end{pmatrix} \begin{pmatrix}
		A_{23} \\
		A_{32}
	\end{pmatrix} \ . \label{boi2}
\end{align}
This $S$-matrix only accounts for the mixed-flavour entries of the full two-body $S$-matrix. If we want the full $S$-matrix, we have to repeat the above computations with the wavefunctions
\begin{align*}
	|\psi_{22} \rangle &=\sum_{1 \leq n_1 < n_2 \leq L} \left[ A_{22} e^{i(p_1  n_1 + p_2  n_2)} q_3^{n_1+n_2} + \tilde{A}_{22} e^{i(p_1  n_2 + p_2  n_1)} q_3^{n_1+n_2} \right] S^{2,+}_{n_1} S^{2,+}_{n_2} |0\rangle \ , \\
	|\psi_{33} \rangle &=\sum_{1 \leq n_1 < n_2 \leq L} \left[ A_{33}  \frac{e^{i(p_1  n_1 + p_2  n_2)} }{q_2^{n_1+n_2}} + \tilde{A}_{33}  \frac{e^{i(p_1  n_2 + p_2  n_1)}}{q_2^{n_1+n_2}} \right] S^{3,+}_{n_1} S^{3,+}_{n_2} |0\rangle \ ,
\end{align*}
which give us the same dispersion relation and $S$-matrix as a regular $\mathfrak{su}(2)$ spin chain. Thus, the complete $S$-matrix takes the form
\begin{equation}
	S(p_2, p_1)=\begin{pmatrix}
		-\frac{1-2 e^{i p_2} + e^{i (p_1+p_2)}}{1-2 e^{i p_1} + e^{i (p_1+p_2)}} & 0 & 0 & 0 \\
		0 & \frac{1}{q_1 q_2 q_3} \frac{e^{i p_2}-e^{i p_1}}{1-2 e^{i p_1} + e^{i (p_1+p_2)}} & \frac{-(1-e^{i p_1})(1-e^{i p_2})}{1-2 e^{i p_1} + e^{i (p_1+p_2)}} & 0 \\
		0 & \frac{-(1-e^{i p_1})(1-e^{i p_2})}{1-2 e^{i p_1} + e^{i (p_1+p_2}} &  q_1 q_2 q_3 \frac{e^{i p_2}-e^{i p_1}}{1-2 e^{i p_1} + e^{i (p_1+p_2)}} & 0 \\
		0 & 0 & 0 & -\frac{1-2 e^{i p_2} + e^{i (p_1+p_2)}}{1-2 e^{i p_1} + e^{i (p_1+p_2)}}
	\end{pmatrix} \ .
\end{equation}

Finally, as we are working with a closed spin chain, we have to impose that our wavefunction is periodic. If we write the wavefunction as $|\psi_{ij} \rangle = \sum_{n_1 < n_2} |\psi_{ij} (n_1 , n_2) \rangle$, the periodicity condition can be expressed as $|\psi_{ij} (n_1 , n_2) \rangle=|\psi_{ij} (n_2 , n_1+L) \rangle$. Substituting our ansatz for the case of two excitations, this equation takes the form
\begin{equation}
\label{yo}
	e^{i p_1 L} \begin{pmatrix}
		q_3^L A_{22} \\
		q_3^L A_{23} \\
		q_2^{-L} A_{32} \\
		q_2^{-L} A_{33}
	\end{pmatrix}=\begin{pmatrix}
		\tilde{A}_{22} \\
		\tilde{A}_{23} \\
		\tilde{A}_{32} \\
		\tilde{A}_{33}
	\end{pmatrix}= S (p_1 , p_2)  \begin{pmatrix}
		A_{22} \\
		A_{23} \\
		A_{32} \\
		A_{33}
	\end{pmatrix} \ , \qquad e^{i p_2 L} \begin{pmatrix}
		q_3^L \tilde{A}_{22} \\
		q_2^{-L} \tilde{A}_{23} \\
		q_3^L \tilde{A}_{32} \\
		q_2^{-L} \tilde{A}_{33}
	\end{pmatrix}= \begin{pmatrix}
		A_{22} \\
		A_{23} \\
		A_{32} \\
		A_{33}
	\end{pmatrix}= S (p_2 , p_1)  \begin{pmatrix}
		\tilde{A}_{22} \\
		\tilde{A}_{23} \\
		\tilde{A}_{32} \\
		\tilde{A}_{33}
	\end{pmatrix}\ .
\end{equation}
These equations are usually referred to as \textit{matrix Bethe equations} in the literature.

Thanks to the integrability of the theory, we can extend the matrix Bethe equations to any number of excitations $M$ and $K$
\begin{equation}
	e^{i p_k L} q_3^{(3-f_k)L} q_2^{(2-f_k)L}=S(p_k , p_{k+1}) \dots S(p_k , p_M) S (p_k , p_1) \dots S (p_k , p_{k-1} )\ ,
\end{equation}
where $f_k$ is the flavour of the $k$-th excitation. This way we get a dressing $q_2^{-L}$ if $f_k=3$ and a dressing $q_3^L$ if $f_k=2$. The presence of these additional factors prevents us to use the regular NCBA construction, but it is enough to slightly modify the prescription to make room for them.

As already mentioned, we will be focussing only on the case of $K=1$. To do so, we consider the following ansatz for the wavefunction with a general number of excitations
\begin{equation}
	\left| \psi \right\rangle =\sum_{k=1}^M \sum_{\sigma \in S_M} \psi_k (\sigma) e^{i \sum_j p_{\sigma (j)} n_j} \frac{q_3^{\sum_j n_j}}{(q_2 q_3)^{n_k}} S_{n_1}^{2+} S_{n_2}^{2+} \cdots S_{n_k}^{3+} \cdots S_{n_M}^{2+} |0\rangle\ . \label{wavefunctiongeneralM}
\end{equation}
In this notation, the coefficients of (\ref{M2wavefunction}) are $\psi_2 (Id.)=A_{23}$, $\psi_1 (Id.)=A_{32}$, $\psi_1 (\tau) =\tilde{A}_{32}$ and $\psi_2 (\tau) =\tilde{A}_{23}$, where $\tau$ is the permutation of $1$ and $2$. To solve the matrix Bethe equations acting on this wavefunction we will closely follow the recipe from section II.O of \cite{Sutherland}, with the appropriate modifications due to the twist factors $q_i$.

As we will be dealing with a wavefunction where nearly all the excitations have the same flavour, it is simpler to solve the matrix Bethe equation for an ``undressed'' $S$-matrix that has the upper-left entry equal to $1$,
\begin{align}
	\lambda_k q_3^{(3-f_k)L} q_2^{(2-f_k)L} \left| \psi \right\rangle &=s(p_k , p_{k+1}) \dots s(p_k , p_M) s (p_k , p_1) \dots s (p_k , p_{k-1} ) \left| \psi \right\rangle \ , \\
	s(p_i,p_j) =-\frac{1-2 e^{i p_j} + e^{i (p_i+p_j)}}{1-2 e^{i p_i} + e^{i (p_i+p_j)}} S(p_i,p_j) &=\begin{pmatrix}
		1 & 0 & 0 & 0 \\
		0 & \frac{1}{q_1 q_2 q_3} \frac{e^{i p_j}-e^{i p_i}}{1-2 e^{i p_i} + e^{i (p_j+p_i)}} & \frac{(1-e^{i p_j})(1-e^{i p_i})}{1-2 e^{i p_i} + e^{i (p_j+p_i)}} & 0 \\
		0 & \frac{(1-e^{i p_j})(1-e^{i p_i})}{1-2 e^{i p_i} + e^{i (p_j+p_i}} &  q_1 q_2 q_3 \frac{e^{i p_j}-e^{i p_i}}{1-2 e^{i p_i} + e^{i (p_j+p_i)}} & 0 \\
		0 & 0 & 0 & 1
	\end{pmatrix} \ . \notag
\end{align}
In addition, we will be changing from the momentum variable to a rapidity variable defined by the map $e^{ip_i}=\frac{u_i}{u_i+1}$.\footnote{The definition of the rapidity we use is the inverse of the one used in \cite{StaudacherAhn}. We have chosen this definition because our definition of the momentum is also the opposite of the one they use, so we can immediately compare the Bethe equations we obtain with the ones of \cite{StaudacherAhn} without any additional transformation.} We do so because this variable will prove extremely useful in solving the Bethe matrix equation. In terms of these variables, the undressed $S$-matrix reads
\begin{equation}
	s(u_i,u_j)=\begin{pmatrix}
		1 & 0 & 0 & 0 \\
		0 & \frac{1}{q_1 q_2 q_3} \frac{u_j - u_i}{u_j-u_i+1} & \frac{1}{u_j-u_i+1} & 0 \\
		0 & \frac{1}{u_j-u_i+1} &  q_1 q_2 q_3 \frac{u_j - u_i}{u_j-u_i+1} & 0 \\
		0 & 0 & 0 & 1
	\end{pmatrix} \ .
\end{equation}

Let us consider now what happens when we disentangle the product of $S$-matrices ``from inside to outside'' if we focus on the coefficients of the form $\psi_k (Id.)$. As we will be considering nearly exclusively these coefficients, we will drop the argument $(Id.)$ in all the equations that follows unless it is necessary. Starting with the innermost operator action, i.e.
\begin{align*}
s(u_k , u_{k+1}) \dots s(u_k , u_M) s (u_k , u_1) \dots \underbrace{ s (u_k , u_{k-1} ) \left| \psi \right\rangle} \ ,
\end{align*}
and noticing that no other factor of the string of $S$-matrices will affect $\psi_{k-1}$ again, the eigenvector equation becomes the following algebraic equation for $\psi_{k-1}$
\begin{equation}
	q_3^L \lambda_k \psi_{k-1}=\frac{1}{u_{k-1}-u_k+1} \psi_{k} + q_1 q_2 q_3 \frac{u_{k-1}-u_k}{u_{k-1}-u_k+1} \psi_{k-1} \ .
\end{equation}
In addition, we will define the function $\psi^{(1)}_k$ as
\begin{equation}
	\psi^{(1)}_k=\frac{1}{q_1 q_2 q_3} \frac{u_{k-1}-u_k}{u_{k-1}-u_k+1} \psi_{k} +\frac{1}{u_{k-1}-u_k+1} \psi_{k-1} \ .
\end{equation}
A similar reasoning holds for the remaining coefficients $\psi_{l}$, giving us the equations
\begin{gather}
	q_3^L \lambda_k \psi_{k-l-1}=\frac{1}{u_{k-l-1}-u_{k}+1} \psi^{(l-1)}_k + q_1 q_2 q_3 \frac{u_{k-l-1}-u_{k}}{u_{k-l-1}-u_{k}+1} \psi_{k-l-1} \ , \label{Suth1} \\
	\psi^{(l)}_k=\frac{1}{q_1 q_2 q_3} \frac{u_{k-l-1}-u_{k}}{u_{k-l-1}-u_{k}+1} \psi^{(l-1)}_k +\frac{1}{u_{k-l-1}-u_{k}+1} \psi_{k-l-1} \ , \label{Suth2}
\end{gather}
where $k-l-1$ and the other superscripts are all to be understood as integers mod $M$. Notice also that equation (\ref{Suth1}) is valid for all the cases except for $l=M-1$ (which corresponds to $k-l-1$ mod $M=k$) where we have to substitute the $q_3^L $ factor by $q_2^{-L}$. This happens because $f_k=3$ and $f_l=2$ if $l\neq k$ for the part of the wavefunction associated with the coefficient $\psi_k$.

Regarding the other coefficients of the wavefunction, $\psi_k (\sigma)$ for $\sigma$ different from identity, we can obtain those coefficients either by appropriately applying the $S$-matrix to $\psi_k (Id.)$, as is shown in the case of $M=2$ in equation (\ref{yo}), or using the periodicity condition.

If we use (\ref{Suth1}) to eliminate the coefficients $\psi^{(l)}_k$ and $\psi^{(l-1)}_k$ from (\ref{Suth2}), we get an equation for the ratio of coefficients
\begin{displaymath}
	\frac{\psi_{k-l-1}}{\psi_{k-l}}=\frac{1}{q_1 q_2 q_3} \, \frac{q_1 q_2 q_3 +(u_k-u_{k-l}) (q_1 q_2 q_3-q_3^{L} \lambda_k ) }{q_3^{L} \lambda_k + (q_1 q_2 q_3-q_3^{L} \lambda_k ) ( u_k - u_{k-l-1})} \ .
\end{displaymath}
To find the eigenvalue $\lambda_k$, we have to realise that the left-hand side of the equation depends only on the difference $k-l$, while the right-hand side has some terms that depend only on $k$. The ratio can be rewritten as\footnote{In terms of the momenta variables, the previous equation takes a more cumbersome form $\frac{\psi_{k-j-1}}{\psi_{k-j}}=\frac{ (e^{i p_{k-j-1}} -1) [q_1 q_2 q_3 (1-2e^{i p_{k-j}} +e^{i (p_{k-j} + p_k)})  + q_3^L \lambda_k (e^{i p_k} - e^{i p_{k-j}})]}{q_1 q_2 q_3 (e^{i p_{k-j}} -1) [ q_3^L \lambda_k (1-2e^{i p_{k}} +e^{i (p_{k-j-1} + p_k)})  - q_1 q_2 q_3 (e^{i p_k} - e^{i p_{k-j-1}})]}$. This makes the reasoning in this paragraph far from transparent.}
\begin{align}
\frac{\psi_{k-l-1}}{\psi_{k-l}} = \frac{1}{q_1 q_2 q_3} \, \frac{u_{k-l} - \bar{x}}{u_{k-l-1} - \bar{x}+1} \ , \quad \text{ with } \quad \bar{x} -u_k = \frac{q_1 q_2 q_3}{q_1 q_2 q_3- \lambda_k q_3^{L}} \ . \label{Suth}
\end{align}
Thus, the only way the right-hand side of the equation depends only on the difference $k-l$ is if $\bar{x}$ is a constant. If we solve for the eigenvalue $\lambda_k$, we get
\begin{equation}
	q_3^L \lambda_k= q_1 q_2 q_3 \frac{\bar{x} -u_k-1}{\bar{x} -u_k} \ .
\end{equation}

At this point, our argument slightly diverges from the one laid down in \cite{Sutherland}. The original NCBA argument uses equation (\ref{Suth}) to write a general $\psi_l$ in terms of $\psi_1$ and substitutes that expression in (\ref{Suth1}). We cannot do the same due to the twist factors $q_i$, which makes the equation for $\psi_k$ different from the rest. Thus, we are forced to relate a general $\psi_l$ to $\psi_k$. This gives us
\begin{equation}
	\psi_{k-l}=\left[\frac{1}{(q_1 q_2 q_3)^l} \prod_{j=0}^{l-1} \frac{u_{k-j} - \bar{x}}{u_{k-j-1} - \bar{x} +1} \right] \psi_k =q_1 q_2 q_3 \, \frac{u_{k} - \bar{x} +1}{u_{k-l} - \bar{x}} \frac{\psi_k}{\prod_{j=0}^l (q_3^L \lambda_{k-j})} \ .
\end{equation}
In addition, the functions $\psi^{(l)}_k$ take the form
\begin{equation}
	\psi^{(l-1)}_k=q_1 q_2 q_3 \frac{u_{k-l-1} - \bar{x} +1}{u_k - \bar{x}} \psi_{k-l-1} =\frac{\psi_k}{\prod_{j=1}^{l-1} (q_3^L \lambda_{k-j})} \ .
\end{equation}
The final step of the computation is to ``close'' the system of equations by identifying the excitation $M+1$ with the excitation $1$. This amounts to equating the function $\psi^{(M-1)}_k$ with $q_2^{-L} \lambda_k \psi_k$. This identification gives us a constraint to the eigenvalues $\lambda_k$, which translates into a constraint for $\bar{x}$
\begin{equation}
	\frac{\lambda_k}{q_2^L} \prod_{\substack{j=1 \\ j\neq k}}^M q_3^L \lambda_j =1 \Longrightarrow  \frac{({q_2}{q_3})^L}{({q_1}{q_2}{q_3})^{M}}
\prod_{j=1}^M\frac{\bar{x}-u_k}{\bar{x}-u_k-1}=1 \ . \label{AuxBetheEq}
\end{equation}
This equation that the constant $\bar{x}$ has to satisfy is the auxiliary Bethe equation of our model.

Once we have solved the matrix Bethe equations for the undressed $S$-matrix, solving the matrix Bethe equations for the original $S$-matrix is immediate, as they reduce to the algebraic equation
\begin{equation}
	e^{i p_k L}= \lambda_k\prod_{j \neq k} \left( - \frac{1-2 e^{i p_k} + e^{i (p_k+p_j)}}{1-2 e^{i p_j} + e^{i (p_k+p_j)}} \right)
\end{equation}
where $\lambda_k$ are the eigenvalues we have computed above. Substituting them and replacing the momenta by rapidities, the Bethe equations take the form
\begin{equation}
	\frac{q_3^L}{q_1 q_2 q_3} \frac{\bar{x} -u_k}{\bar{x} -u_k-1} \prod_{j \neq k} \frac{u_k - u_j +1}{u_k - u_j -1} = \left( \frac{u_k+1}{u_k} \right)^L \ . \label{BetheEq}
\end{equation}
The fact that both the Bethe equations and the auxiliary Bethe equation we obtained are exactly the same as the ones appearing in \cite{StaudacherAhn} when we identify our $\bar{x}$ with their $v+1$ works as a sanity check.

\subsection{The coalescence of the Bethe states and the structure of generalised eigenvectors}

After computing the eigenvectors for finite values of the twists, we can now use our algorithm to find the generalised eigenvectors of the Hamiltonian (\ref{eq:twistedXXX}) at the exceptional point of large twists. As we have done in previous sections, we will substitute the twist parameters $q_i$ by $\frac{\xi_i}{\epsilon}$ and compute the limit $\epsilon \rightarrow 0$.

We will assume that chain mixing does not affect the naïve identification of Jordan chains, but we will only prove it for some specific cases. Nevertheless, we will compare our results with the results from \cite{StaudacherAhn} and \cite{Ahn:2021emp} to show that they agree.

We will study in detail the particular cases of $M=2$ and $M=3$ before proceeding to analyse the case of arbitrary $M$.

\subsubsection{Two excitations}\label{M2section}

If we want to find how each of the terms appearing in the wavefunction (\ref{M2wavefunction}) behaves in the strongly twisted limit, we need to find how the momenta $p_1$ and $p_2$ and the coefficients $A_{ij}$ and $\tilde{A}_{ij}$ behave in such a limit.

In order to find how the momenta scale with $\epsilon$, we have to study the behaviour of the Bethe equations (\ref{BetheEq}) and the auxiliary Bethe equation (\ref{AuxBetheEq}) in the limit of large twist. This computation was carried out previously in \cite{StaudacherAhn}, so we will just borrow their results. In the limit $\epsilon \rightarrow 0$ the two physical rapidities $u_1$, $u_2$ and the auxiliary rapidity $\bar{x}$ scale as follows
\begin{equation}
	u_1 \approx \epsilon^\alpha u_- \ , \qquad u_2 \approx -1+\epsilon^\alpha u_+ \ , \qquad \bar{x} \approx u_2 + \epsilon^\gamma \hat{v} \ ,
\end{equation}
where $\alpha = \frac{L-3}{L-1}$ and $\gamma = 2 L-6$, while the Bethe equations become
\begin{equation}
	u_-^L = \frac{\xi}{\xi_3^L} (u_- - u_+) \ , \quad (-u_+)^L = \frac{\xi}{\xi_2^L} (u_- - u_+) \ , \quad \hat{v}=-\frac{2 \xi_1^L}{\xi^{L-2}} \ ,
\end{equation}
where $\xi=\xi_1 \xi_2 \xi_3$. At the level of momenta, the behaviour of the rapidities translates into
\begin{equation}
	e^{i p_1} = \frac{u_- \epsilon^\alpha}{u_- \epsilon^\alpha + 1} \approx u_- \epsilon^\alpha \ , \qquad e^{i p_2} = \frac{-1+\epsilon^\alpha u_+}{\epsilon^\alpha u_+} \approx -u_+^{-1} \epsilon^{-\alpha} \ .
\end{equation}
This means that the plane wave factor $e^{i (p_1 n_1 + p_2 n_2)}$ behaves as
\begin{equation}
	e^{i (p_1 n_1 + p_2 n_2)} \sim \epsilon^{\alpha (n_1 - n_2)} \ ,
\end{equation}
which depends only on the relative position of the excitations. This was not unexpected, as the state is an eigenvector of the shift operator $U$.

The next step is to find how the coefficients $A_{ij}$ scale with $\epsilon$. As usual, we are free to normalise the wavefunction to the value that better suits us, so we will be choosing $A_{23}=1$ for simplicity. For the case of $M=2$, equation (\ref{Suth}) takes the form
\begin{align*}
\frac{\psi_1 (Id.)}{\psi_2 (Id.)}=\frac{A_{32}}{A_{23}} = \frac{1}{q_1 q_2 q_3} \frac{u_2 - \bar{x}}{u_1-\bar{x}+1} \ .
\end{align*}
For the coefficients with a tilde, we can use equation~(\ref{boi2}) written in terms of rapidities instead of momenta
\begin{equation}
	\begin{pmatrix}
		\tilde{A}_{23} \\
		\tilde{A}_{32}
	\end{pmatrix} =  \begin{pmatrix}
\frac{1}{q_1 q_2 q_3} \frac{u_2 - u_1}{u_2-u_1+1} & \frac{1}{u_2-u_1+1} \\
\frac{1}{u_2-u_1+1} & - q_1 q_2 q_3 \frac{u_2 - u_1}{u_2-u_1+1}	
	\end{pmatrix} \begin{pmatrix}
		A_{23} \\
		A_{32}
	\end{pmatrix} \ .
\end{equation}
Substituting the behaviour of the rapidities when $\epsilon$ approaches zero, we find
\begin{align}
\begin{pmatrix}
		A_{23} \\
		A_{32}
	\end{pmatrix}&= \begin{pmatrix}
 1 \\
 \frac{1}{q_1 q_2 q_3} \, \frac{-\hat{v} \epsilon^\gamma}{2+(u_- - u_+) \epsilon^\alpha - \hat{v} \epsilon^\gamma} 
	\end{pmatrix} \approx  \begin{pmatrix}
 1 \\
 -\frac{\hat{v}}{2 \xi} \epsilon^{\gamma+3}
	\end{pmatrix} \notag \\
	\begin{pmatrix}
\tilde{A}_{32}
 \\
 \tilde{A}_{23}
	\end{pmatrix}&= \begin{pmatrix}
 \frac{\epsilon ^{-\alpha
   }+u_--u_+}{q_1 q_2 q_3 u_--q_1 q_2 q_3
   u_+}+\frac{-\hat{v} \epsilon
   ^{\gamma -\alpha
   }}{q_1 q_2 q_3 (u_--u_+)
   \left(2-u_- \epsilon ^{\alpha
   }+u_+ \epsilon ^{\alpha
   }-\hat{v} \epsilon ^{\gamma
   }\right)} \\
 \frac{\epsilon ^{-\alpha
   }}{u_--u_+}-\frac{-\hat{v} \epsilon ^{\gamma -\alpha} (-1+u_- \epsilon^\alpha -u_+ \epsilon^\alpha)}{(u_--u_+)
   \left(2-u_- \epsilon ^{\alpha
   }+u_+ \epsilon ^{\alpha
   }-\hat{v} \epsilon ^{\gamma
   }\right)}
	\end{pmatrix} \approx  \begin{pmatrix}
 \frac{1}{\xi (u_- - u_+)} \epsilon^{3-\alpha} \\
 \frac{1}{u_- - u_+}  \epsilon^{-\alpha}
	\end{pmatrix}
\end{align}

If we put everything together, we can find that each term that appears in the wavefunction (\ref{M2wavefunction}) behaves as
\begin{align*}
	&A_{23} e^{i(p_1  n_1 + p_2  n_2)} \frac{q_3^{n_1}}{q_2^{n_2}} \sim \epsilon^{(1-\alpha) (n_2-n_1)} \ , & 	&\tilde{A}_{23} e^{i(p_1  n_2 + p_2  n_1)} \frac{q_3^{n_1}}{q_2^{n_2}} \sim \epsilon^{(1+\alpha)(n_2-n_1)-\alpha} \ , \\
	&A_{32} e^{i(p_1  n_1 + p_2  n_2)} \frac{q_3^{n_2}}{q_2^{n_1}} \sim \epsilon^{\gamma + 3 - (1+\alpha)(n_2-n_1)} \ , & 	&\tilde{A}_{32} e^{i(p_1  n_2 + p_2  n_1)} \frac{q_3^{n_2}}{q_2^{n_1}} \sim \epsilon^{3-\alpha-(1-\alpha)(n_2-n_1)} \ ,
\end{align*}
where we should keep in mind that $n_2 - n_1 \in \{1, \cdots, L-1\}$, as we need to take into account the restriction $1 \leq n_1 < n_2 \leq L$ from the wavefunction. With this restriction in mind, we can see that the term associated with the coefficient $A_{23}$ always dominates over the term associated with $\tilde{A}_{23}$ (as we can check that $(1-\alpha) (n_2-n_1) < (1+\alpha)(n_2-n_1)-\alpha$ for all the values we are interested in), and the term associated with the coefficient $\tilde{A}_{32}$ always dominates over the term associated with $A_{32}$ (as, similarly, $3-\alpha-(1-\alpha)(n_2-n_1)<\gamma + 3 - (1+\alpha)(n_2-n_1)$ for all the values we are interested in). Thus, from now on, we will work with the wavefunction (\ref{M2wavefunction}) as if the contributions associated with the coefficients $\tilde{A}_{23}$ and $A_{32}$ were not there.

At this point, we have enough information about the wavefunction $|\psi_{23} \rangle$ to apply our algorithm. We start by looking at how $|\psi_{23} \rangle$ behaves when $\epsilon$ approaches zero
\begin{align*}
	|\psi_{23} \rangle &\approx \sum_{n_1 < n_2} \left[ A_{23} e^{i(p_1  n_1 + p_2  n_2)} \frac{q_3^{n_1}}{q_2^{n_2}} S^{2,+}_{n_1} S^{3,+}_{n_2} + \tilde{A}_{32} e^{i(p_1  n_2 + p_2  n_1)} \frac{q_3^{n_2}}{q_2^{n_1}} S^{3,+}_{n_1} S^{2,+}_{n_2} \right] |0\rangle \\
	&\approx \sum_{n_1 < n_2} \left[ \frac{(\xi_3 u_-)^{n_1}}{(-\xi_2 u_+)^{n_2}} \epsilon^{(1-\alpha ) (n_2 - n_1)} S^{2,+}_{n_1} S^{3,+}_{n_2} + \frac{1}{\xi (u_- - u_+)} \frac{(\xi_3 u_-)^{n_2}}{(-\xi_2 u_+)^{n_1}} \epsilon^{3-\alpha- (1-\alpha)(n_2 - n_1)} S^{3,+}_{n_1} S^{2,+}_{n_2} \right] |0\rangle \ .
\end{align*}
Using that $\alpha$ fulfils the relation $3-\alpha- (1-\alpha)(L-1)=1-\alpha$, we see that the terms that dominate take the form
\begin{displaymath}
	\frac{|\psi_{23} \rangle}{ \epsilon^{1-\alpha}} \approx \left[\sum_{n_1=1}^{L-1} \frac{-1}{\xi_2 u_+} \left(-\frac{\xi_3 u_-}{\xi_2 u_+}\right)^{n_1} S^{2,+}_{n_1} S^{3,+}_{n_1+1} \right] |0\rangle + \frac{1}{\xi (u_- - u_+)} \frac{(\xi_3 u_-)^{L}}{(-\xi_2 u_+)^{1}} S^{3,+}_{1} S^{2,+}_{L} |0\rangle \ .
\end{displaymath}
Using the Bethe equation $(-\xi_2 u_+)^L= \xi (u_- - u_+) $, we can combine the two terms into one, giving us
\begin{equation}
	\frac{|\psi_{23} \rangle}{ \epsilon^{1-\alpha}} \approx \frac{-1}{\xi_2 u_+} \sum_{n_1=1}^{L} \left(-\frac{\xi_3 u_-}{\xi_2 u_+}\right)^{n_1} S^{2,+}_{n_1} S^{3,+}_{n_1+1} |0\rangle = |\psi^{(1)} (u_-/u_+) \rangle \ ,
\end{equation}
where the position of the operators $S^{i,+}_n$ should be understood modulo $L$. We can see that this is indeed the locked state studied in \cite{StaudacherAhn}.

Using the results from the previous sections, we can claim that the states $|\psi^{(1)} (u_-/u_+) \rangle$ are eigenstates of the Hamiltonian (\ref{eq:twistedXXX}) at its exceptional point, that is, they will be eigenstates of the Hamiltonian $\mathbf{\hat{H}}_{\xi_1 , \xi_2 , \xi_3}$. Notice that these states do not depend on the rapidities $u_+$ and $u_-$ separately, but on the combination $\frac{u_-}{u_+}$, which is equivalent to say that they depend only on the total momentum $p_1+p_2$. One might claim that this is not true due to the factor $\frac{1}{\xi_2 u_+}$, but it is nothing more than a normalisation factor that we can eliminate.

The next step will be to compute the vectors $w_{ij}^{(1)}$ defined in equation (\ref{recipeextra}). As we commented in section \ref{howtolimit}, we take two vectors that coalesce to the same vector at the exceptional point, find a linear combination that becomes zero at the exceptional point, and then compute the limit of this linear combination divided by its norm (or, at least, by $\epsilon$ to the correct power such that it gives us a finite non-vanishing limit).

As we have just shown, two vectors $|\psi_{23} (p_1 , p_2) \rangle$ and $|\psi_{23} (p'_1 , p'_2) \rangle$ coalesce to the same eigenvector $|\psi^{(1)} (u_-/u_+) \rangle$ if their total momentum coincide, i.e., if $p_1 + p_2=p'_1 + p'_2$. In addition, by writing them explicitly we can check that
\begin{equation}
	e^{i p_1} |\psi_{23} (p_1 , p_2) \rangle - e^{i p'_1}|\psi_{23} (p'_1 , p'_2) \rangle =0 \ .
\end{equation}
Indeed, if we study how this linear combination behaves when $\epsilon$ approaches zero, we find that the terms with $n_2=n_1+1$ of $e^{i p_1} |\psi_{23} (p_1 , p_2) \rangle$ cancel with those from $e^{i p'_1}|\psi_{23} (p'_1 , p'_2) \rangle$. Therefore, the leading contribution in this linear combination takes the form
\begin{align*}
	&e^{i p_1} |\psi_{23} (p_1 , p_2) \rangle - e^{i p'_1}|\psi_{23} (p'_1 , p'_2) \rangle \approx { \epsilon^{2-\alpha}} \sum_{n_1=1}^{L-2} \left[ \frac{u_-}{(-\xi_2 u_+)^2} \left(-\frac{\xi_3 u_-}{\xi_2 u_+}\right)^{n_1} - \frac{u'_-}{(-\xi_2 u'_+)^2} \left(-\frac{\xi_3 u'_-}{\xi_2 u'_+}\right)^{n_1} \right] \cdot \\
	&\cdot S^{2,+}_{n_1} S^{3,+}_{n_1+2} |0\rangle + \epsilon^{3- (1-\alpha)(L-2)} \left[ \frac{u_-}{\xi (u_- - u_+)} \frac{(\xi_3 u_-)^{L}}{(-\xi_2 u_+)^{2}} - \frac{u'_-}{\xi (u'_- - u'_+)} \frac{(\xi_3 u'_-)^{L}}{(-\xi_2 u'_+)^{2}} \right] S^{3,+}_{2} S^{2,+}_{L} |0\rangle + \\
	&+\epsilon^{3- (1-\alpha)(L-2)} \left[ \frac{u_-}{\xi (u_- - u_+)} \frac{(\xi_3 u_-)^{L-1}}{(-\xi_2 u_+)^{1}} - \frac{u'_-}{\xi (u'_- - u'_+)} \frac{(\xi_3 u'_-)^{L-1}}{(-\xi_2 u'_+)^{1}} \right] S^{3,+}_{1} S^{2,+}_{L-1} |0\rangle \ .
\end{align*}
Using the explicit value of $\alpha$, we can check that $3- (1-\alpha)(L-2)=2-\alpha$. If we now use the condition that the total momentum of both states is the same, that is, $\frac{u_-}{u_+}=\frac{u'_-}{u'_+}$, and the Bethe equations, we can rewrite this expression into a more compact form
\begin{equation}
	\frac{e^{i p_1} |\psi_{23} (p_1 , p_2) \rangle - e^{i p'_1}|\psi_{23} (p'_1 , p'_2) \rangle}{ \epsilon^{2-\alpha}} \approx \frac{u_-}{u_+} \, \frac{u'_+ - u_+}{\xi_2^2 u_+ u'_+} \sum_{n_1=1}^{L} \left(-\frac{\xi_3 u_-}{\xi_2 u_+}\right)^{n_1} S^{2,+}_{n_1} S^{3,+}_{n_1+2} |0\rangle = |\psi^{(2)} (u_-/u_+) \rangle \ .
\end{equation}

Hence, using the results from the previous sections, we can claim that the states $|\psi^{(2)} (u_+/u_-) \rangle$ will be generalised eigenstates of rank 2 of the Hamiltonian (\ref{eq:twistedXXX}) at its exceptional point, that is, they will be generalised eigenstates of rank 2 of the Hamiltonian $\mathbf{\hat{H}}_{\xi_1 , \xi_2 , \xi_3}$. Similarly to $|\psi^{(1)} (u_-/u_+) \rangle$, these states do not depend on the two momenta separately, but only thought the total momentum $\frac{u_-}{u_+}$.

We can repeat the process again with the combination
\begin{equation}
	\frac{e^{i (p_1 + p'_1)}}{e^{i p'_1} - e^{i p_1}} \left( e^{i p_1} |\psi_{23} (p_1 , p_2) \rangle - e^{i p'_1}|\psi_{23} (p'_1 , p'_2) \rangle \right) - \frac{e^{i (p_1 + p^{\prime \prime}_1)}}{e^{i p^{\prime \prime}_1} - e^{i p_1}} \left( e^{i p_1} |\psi_{23} (p_1 , p_2) \rangle - e^{i p^{\prime \prime}_1}|\psi_{23} (p^{\prime \prime}_1 , p^{\prime \prime}_2) \rangle \right)
\end{equation}
which, after some algebra, we can find that it behaves as
\begin{equation}
	\text{const. } \epsilon^{3-\alpha} \sum_{n_1=1}^{L} \left(-\frac{\xi_3 u_-}{\xi_2 u_+}\right)^{n_1} S^{2,+}_{n_1} S^{3,+}_{n_1+3} |0\rangle = \epsilon^{3-\alpha} |\psi^{(3)} (u_-/u_+) \rangle \ .
\end{equation}
This indicates us that $|\psi^{(3)} (u_-/u_+) \rangle $ is a generalised eigenvector of rank 3 of $\mathbf{\hat{H}}_{\xi_1 , \xi_2 , \xi_3}$ up to a possible additional contribution proportional to $|\psi^{(2)} (u_-/u_-) \rangle $, as we commented in section~\ref{distinguishableJB} and \ref{howtolimit}.

If we continue with this process, we find that the generalised eigenvector of rank $k$ of $\mathbf{\hat{H}}_{\xi_1 , \xi_2 , \xi_3}$ with $M=2$ and $K=1$ takes the form
\begin{equation}
	 |\psi^{(k)} (u_-/u_+) \rangle =\text{const. }  \sum_{n_1=1}^{L} \left(-\frac{\xi_3 u_-}{\xi_2 u_+}\right)^{n_1} S^{2,+}_{n_1} S^{3,+}_{n_1+k} |0\rangle \ .
\end{equation}
As we see, this expression is only meaningful if $k \: \in \: \{1, \cdots L-1 \}$. In addition, we find exactly one vector for a given value of the total momentum at any step of the process. Thus, we can claim with some certainty that the sector of $M=2$ and $K=1$ has Jordan chains of size $L-1$ for every allowed value of the total momentum. This result is in agreement with the numerical results from \cite{StaudacherAhn} and the combinatorial results from \cite{Ahn:2021emp}.

We want to stress that this is only a claim at this point, as we cannot be sure that we are not mislead by chain mixing. To show that this is not the case, we will follow the procedure we described at the end of section \ref{nondistinguishableJB}. Consider the adjoint of the strongly twisted Hamiltonian  
\begin{align}
	& & \mathbf{\hat{H}}_{\xi_1 , \xi_2 , \xi_3}^\dagger &= \sum_{l=1}^L (\hat{\mathbb{P}}^{l,l+1})^\dagger \ , \\
\hat{\mathbb{P}}^\dagger \,|12\rangle &=\xi_3^*\, |21\rangle \ , & \hat{\mathbb{P}}^\dagger \,|23\rangle &=\xi_1^*\, |32\rangle \ , & \hat{\mathbb{P}}^\dagger \, |31\rangle &=\xi_2^*\, |13\rangle \ .
\end{align}
According to corollary \ref{lowesttohighest}, the generalised eigenvector of rank 1 we have computed, $|\psi^{(1)} (u_-/u_+) \rangle$, behaves like a generalised eigenvector of maximal rank of $\mathbf{\hat{H}}_{\xi_1 , \xi_2 , \xi_3}^\dagger$. In fact, we can check that\footnote{Notice that not all the left generalised eigenvectors are orthogonal to the left generalised eigenvector of rank 1, $|\psi^{(L-1)} (u_-/u_+) \rangle$. This is because the corollary \ref{lowesttohighest} only requires orthogonality with respect to the eigenvector we are starting with, not with the one associated with the generalised eigenvector of highest rank we are computing.}
\begin{align*}
	&\mathbf{\hat{H}}_{\xi_1 , \xi_2 , \xi_3}^\dagger S_{n}^{2,+} S_{n+1}^{3,+} |0 \rangle= \xi_3^* S_{n-1}^{2,+} S_{n+1}^{3,+} |0 \rangle + \xi_1^* S_{n}^{3,+} S_{n+1}^{2,+} |0 \rangle + \xi_2^* S_{n}^{2,+} S_{n+3}^{3,+} |0 \rangle \ , \\
	&\mathbf{\hat{H}}_{\xi_1 , \xi_2 , \xi_3}^\dagger |\psi^{(1)} (u_-/u_+) \rangle= \rho (\xi_2^* + \xi_3^*) |\psi^{(2)} (u_-/u_+) \rangle + \rho'\xi_1^* |\psi^{(L-1)} (u_-/u_+) \rangle \ , \\
	&\mathbf{\hat{H}}_{\xi_1 , \xi_2 , \xi_3}^\dagger \mathbf{\hat{H}}_{\xi_1 , \xi_2 , \xi_3}^\dagger |\psi^{(1)} (u_-/u_+) \rangle \propto (\xi_2^* + \xi_3^*)^2 |\psi^{(3)} (u_-/u_+) \rangle \ , \\
	&(\mathbf{\hat{H}}_{\xi_1 , \xi_2 , \xi_3}^\dagger )^k |\psi^{(1)} (u_-/u_+) \rangle \propto (\xi_2^* + \xi_3^*)^k |\psi^{(k+1)} (u_-/u_+) \rangle \ , \\
	&(\mathbf{\hat{H}}_{\xi_1 , \xi_2 , \xi_3}^\dagger )^{L-1} |\psi^{(1)} (u_-/u_+) \rangle \propto \mathbf{\hat{H}}_{\xi_1 , \xi_2 , \xi_3}^\dagger (\xi_2^* + \xi_3^*)^{L-2} |\psi^{(L-1)} (u_-/u_+) \rangle =0 \ ,
\end{align*}
where $\rho$ and $\rho ' $ are two constants that are not important for our discussion. Notice that the generalised eigenvector of highest rank has the form of a locked state, but with the 2's to the right of the 3, instead of to its left. These states are called \emph{anti-locked states}, and they were shown to be generalised eigenvectors of highest rank in \cite{StaudacherAhn}.

As we now have the generalised eigenvector of highest rank of the Jordan chain, we can apply the definition of generalised eigenvector to compute the remaining generalised eigenvectors of the chain
\begin{equation}
	\mathbf{\hat{H}}_{\xi_1 , \xi_2 , \xi_3}^k |\psi^{(L-1)} (u_-/u_+) \rangle = \rho^{\prime \prime} (\xi_2 + \xi_3)^k |\psi^{(L-k-1)} (u_-/u_+) \rangle + \rho^{\prime \prime \prime} \xi_1 \delta_{k,1} |\psi^{(1)} (u_-/u_+) \rangle \ ,
\end{equation}
where $\rho^{\prime \prime}$ and $\rho^{\prime \prime \prime}$ are two constants that are not important for our discussion. Thus, as long as $\xi_2 + \xi_3 \neq 0$, we are able to act $L-2$ times with the strongly twisted Hamiltonian on $|\psi^{(L-1)} (u_-/u_+) \rangle$. This means that our claim of having Jordan cells of size $L-1$ is correct provided $\xi_2 + \xi_3 \neq 0$. On the other hand, if $\xi_2 + \xi_3 = 0$ and $\xi_1 \neq 0$, our claim is wrong and we actually have a Jordan chain of size $2$ and $L-3$ eigenvectors. In addition, if $\xi_2 + \xi_3 = 0$ and $\xi_1 = 0$, there are no non-trivial Jordan chains and we have $L-1$ eigenvectors.

The issue of chain mixing in this setting was also studied in \cite{Ahn:2021emp}. The authors conjectured that chain mixing does not affect these results for $K=1$, but they were not able to provide a proof of this statement. Nevertheless, they checked that their conjecture numerically up to $L=30$ and $M=6$.

\subsubsection{Three excitations}

Let us move now to the case of $M=3$. The analysis follows the same structure as in the $M=2$ case, but it presents some new peculiarities that are worth mentioning before tackling the sector with a general number of excitations.

As in the previous case, we start by studying the behaviour of the Bethe equations (\ref{BetheEq}) and the auxiliary Bethe equation (\ref{AuxBetheEq}) in the limit of large twist. Although this computation was carried out previously in \cite{StaudacherAhn}, we need the next-to-leading order contribution for $u_1$ and $u_2$, which were not computed there. In the limit $\epsilon \rightarrow 0$ the three physical rapidities $u_1$, $u_2$, $u_3$ and the auxiliary rapidity $\bar{x}$ scale as follows
\begin{equation}
	u_1 \rightarrow \epsilon^\alpha u_- + \epsilon^{2 \alpha} \tilde{u}_- \ , \qquad u_2 \rightarrow \epsilon^\alpha u'_-+ \epsilon^{2 \alpha} \tilde{u}'_- \ , \qquad u_3 \rightarrow -1+\epsilon^\beta u_+ \ , \qquad \bar{x} \rightarrow u_3 + \epsilon^\gamma \hat{v} \ ,
\end{equation}
where $\alpha = \frac{L-M-1}{L-M+1}$, $\beta = \frac{L-3M-3}{L-M+1}$ and $\gamma = 2 L-3M$, while the Bethe equations become
\begin{align*}
	&u_-^L =(-1)^{M-1} \frac{\xi}{\xi_3^L} u_+ \ , \quad \frac{ u_-^2 - \tilde{u}_-}{u_-^{L+1}} = 2 \frac{\xi_3^L}{\xi} \frac{\tilde{u}_- - \tilde{u}'_-}{u_+}  \ , \\
	&(u'_-)^L=(-1)^{M-1} \frac{\xi}{\xi_3^L} u_+ \ , \quad \frac{ u_-^{\prime 2} - \tilde{u}'_-}{(u'_-)^{L+1}} = 2 \frac{\xi_3^L}{\xi} \frac{\tilde{u}'_- - \tilde{u}_-}{u_+}  \ , \\
	&-(-u_+)^{L-M+1} = \frac{\xi_3^L}{2^{M-1} \xi} \hat{v} \ , \quad \hat{v}=-\frac{2^{M-1} \xi_1^L}{\xi^{L-M}} \ ,
\end{align*}
where $\xi=\xi_1 \xi_2 \xi_3$. At the level of momenta, the behaviour of the rapidity translate into
\begin{equation}
	e^{i p_1} \approx u_- \epsilon^\alpha + (\tilde{u}_- - u_-^2) \epsilon^{2\alpha} \ , \qquad e^{i p_2} \approx u'_- \epsilon^\alpha + (\tilde{u}'_- - u_-^{\prime 2}) \epsilon^{2\alpha} \ , \qquad e^{i p_3} \approx -u_+^{-1} \epsilon^{-\beta} \ .
\end{equation}
Substituting this result into the plane wave factor $\exp [i (p_1 n_1 + p_2 n_2 + p_3 n_3)]$, we find that it behaves as
\begin{equation}
	e^{i (p_1 n_1 + p_2 n_2 + p_3 n_3)} \frac{q_3^{n_1} q_3^{n_2}}{q_2^{n_3}} \sim \epsilon^{(\alpha-1) (n_1 + n_2) - (\beta -1) n_3}=\epsilon^{(\alpha-1) [(n_1 -n_3) + (n_2 -n_3)] } \ ,
\end{equation}
where we have used that $(M-K) (\alpha-1)=K(\beta -1)$. As it happened in the $M=2$ case, this quantity depends only on the differences between the positions of excitations of type $2$ and the excitation of type $3$ because our wavefunction is an eigenstate of the shift operator $U$.

Let us move now to the coefficients $\psi_k (\sigma)$. First, we will fix $\psi_3 (Id.)=1$ and use equation~(\ref{Suth}) to fix $\psi_1 (Id.)$ and $\psi_2 (Id.)$
\begin{equation}
	\frac{\psi_{2}(Id.)}{\psi_{3}(Id.)} = \frac{1}{q_1 q_2 q_3} \, \frac{u_{3} - \bar{x}}{u_{2} - \bar{x}+1} \ , \quad   \frac{\psi_{1}(Id.)}{\psi_{2}(Id.)} = \frac{1}{q_1 q_2 q_3} \, \frac{u_{2} - \bar{x}}{u_{1} - \bar{x}+1} \ .
\end{equation}
In addition, we can easily construct the coefficients $\psi_{i} ((jk))$, with $i$, $j$, $k$ different integers. As the permutation $(jk)$ affects two excitations with the label $2$, these excitations scatter with the usual $\mathfrak{su}(2)$ $S$-matrix, and we have $\psi_{i} ((jk))=S_{22} (p_j , p_k)  \psi_{i} (Id.)=-\frac{1-2 e^{i p_j} + e^{i (p_j+p_k)}}{1-2 e^{i p_k} + e^{i (p_j+p_k)}} \psi_{i} (Id.)$. This is enough to fix a total of six coefficients of the wavefunction.

To find the remaining coefficients, we can use the periodicity condition to relate a term of the form $S^{2 +}_{n_1}  S^{2 +}_{n_2}  S^{3 +}_{n_3} $ with a term of the form $S^{3 +}_{n'_1}  S^{2 +}_{n'_2}  S^{2 +}_{n'_3} $, and, if we apply it again, to term of the form $S^{2 +}_{n^{\prime \prime}_1}  S^{3 +}_{n^{\prime \prime}_2}  S^{2 +}_{n^{\prime \prime}_3} $. This implies that, for example
\begin{equation}
	\psi_1 ((321)) = \psi_3 (Id.)  \frac{e^{i p_3 L}}{q_2^L} \ , \quad \psi_2 ((123)) = \psi_3 (Id.) e^{i (p_2+p_3) L} \frac{q_3^L}{q_2^L} \ .
\end{equation}
As we had a total of six coefficients, and we can obtain two other coefficients from a given one through these relations, this is enough to completely fix all the coefficients of the wavefunction.

Putting everything together, we can see that the coefficients that are associated with $\psi_3 (Id.)$ are the ones that dominate with respect to the other contributions. Consider, for example, the part of the wavefunction associated with the raising operators $S^{2 +}_{n_1}  S^{2 +}_{n_2}  S^{3 +}_{n_3} $
\begin{align*}
	\psi_3 (Id.) e^{i (p_1 n_1 + p_2 n_2 + p_3 n_3)} \frac{q_3^{n_1 + n_2}}{q_2^{n_3}} &\sim \psi_3 ((12)) e^{i (p_2 n_1 + p_1 n_2 + p_3 n_3)} \frac{q_3^{n_1 + n_2}}{q_2^{n_3}} \sim \epsilon^{(\alpha-1) [(n_1 -n_3) + (n_2 -n_3)] } \ , \\
	\psi_3 ((123)) e^{i (p_2 n_1 + p_3 n_2 + p_1 n_3)} \frac{q_3^{n_1 + n_2}}{q_2^{n_3}} &\sim \psi_1 (Id.) \frac{q_2^L}{e^{ip_1 L}} \epsilon^{\alpha (n_1+n_3) - \beta n_2} \epsilon^{n_1 + n_2 - n_3} \\
	&\sim \epsilon^{\gamma -3} \epsilon^{-(1+\alpha) L} \epsilon^{(\alpha-1) [(n_1 -n_3) + (n_2 -n_3)]+ (\alpha + \beta) (n_3 - n_2)} \ , \\
	\psi_3 ((13)) e^{i (p_3 n_1 + p_2 n_2 + p_1 n_3)} \frac{q_3^{n_1 + n_2}}{q_2^{n_3}} &\sim \psi_1 (Id.) S_{22} (p_3 , p_2) \frac{q_2^L}{e^{ip_1 L}} \epsilon^{\alpha (n_2+n_3)-\beta n_1 } \epsilon^{n_1 + n_2 - n_3} \\
	&\sim \epsilon^{\gamma -3} \epsilon^{-\beta} \epsilon^{-(1+\alpha) L} \epsilon^{(\alpha-1) [(n_1 -n_3) + (n_2 -n_3)]+ (\alpha + \beta) (n_3 - n_1)} \ , \\
	\psi_3 ((321)) e^{i (p_3 n_1 + p_1 n_2 + p_2 n_3)} \frac{q_3^{n_1 + n_2}}{q_2^{n_3}} &\sim \psi_2 (Id.) \frac{q_2^L}{e^{i (p_1 + p_2) L} q_3^L} \epsilon^{\alpha (n_2+n_3)-\beta n_1 } \epsilon^{n_1 + n_2 - n_3} \\
	&\sim \epsilon^3 \epsilon^{-2\alpha L} \epsilon^{(\alpha-1) [(n_1 -n_3) + (n_2 -n_3)]+ (\alpha + \beta) (n_3 - n_1)} \ , \\
	\psi_3 ((23)) e^{i (p_1 n_1 + p_3 n_2 + p_2 n_3)} \frac{q_3^{n_1 + n_2}}{q_2^{n_3}} &\sim \psi_2 (Id.) \frac{q_2^L}{e^{i (p_1 + p_2) L} q_3^L} S_{22} (p_1 , p_3) \epsilon^{\alpha (n_1+n_3) - \beta n_2} \epsilon^{n_1 + n_2 - n_3} \\
	&\sim \epsilon^3 \epsilon^{-2\alpha L} \epsilon^{\beta} \epsilon^{\alpha (n_1+n_3) - \beta n_2} \epsilon^{n_1 + n_2 - n_3} \ .
\end{align*}
We should be careful, because this analysis overlooks one detail. In the $\epsilon\rightarrow 0$ limit the momenta $p_1$ and $p_2$ become equal. This means that the leading contribution of $\psi_3 (Id.) e^{i (p_1 n_1 + p_2 n_2 + p_3 n_3)} \frac{q_3^{n_1 + n_2}}{q_2^{n_3}}$ perfectly cancels the leading contribution of $\psi_3 ((12)) e^{i (p_2 n_1 + p_1 n_2 + p_3 n_3)} \frac{q_3^{n_1 + n_2}}{q_2^{n_3}}$, as $S_{22}(p,p)=-1$ if the two excitations have the same flavour. Thankfully, the next-to-leading contribution to the rapidities $u_1$ and $u_2$ makes them different, giving us an additional $\epsilon^\alpha$ factor compared to our previous naïve computation
\begin{equation}
	\psi_3 (Id.) e^{i (p_1 n_1 + p_2 n_2 + p_3 n_3)} \frac{q_3^{n_1 + n_2}}{q_2^{n_3}} + \psi_3 ((12)) e^{i (p_2 n_1 + p_1 n_2 + p_3 n_3)} \frac{q_3^{n_1 + n_2}}{q_2^{n_3}}\sim \epsilon^\alpha \epsilon^{(\alpha-1) [(n_1 -n_3) + (n_2 -n_3)] } \ .
\end{equation}
All in all, we find a situation analogous to the $M=2$ case: the contributions associated with $\psi_3 (Id.)$ and $\psi_3 ((12))$ dominate over the other four. Thus, from now on, we will work with the wavefunction as if only the contributions associated with the coefficient $\psi_3 (Id.)$ were the only ones present.

At this point, we have enough information about the wavefunction (\ref{wavefunctiongeneralM}) to apply our algorithm. We start by looking at how $|\psi \rangle$ behaves when $\epsilon$ approaches zero
\begin{align*}
	|\psi \rangle &\approx \sum_{\substack{n_1 < n_2 < n_3 \\ 1\leq n_i \leq L}}^L \left[ \left( \psi_3 (Id.) e^{i(p_1  n_1 + p_2  n_2)} + \psi_3 ((12)) e^{i(p_2  n_1 + p_1  n_2)} \right) e^{ip_3 n_3} \frac{q_3^{n_1+n_2}}{q_2^{n_3}} S^{2,+}_{n_1} S^{2,+}_{n_2} S^{3,+}_{n_3} \right. \\
	&+ \left( \psi_1 ((321)) e^{i(p_1  n_2 + p_2  n_3)} + \psi_1 ((13)) e^{i(p_2  n_2 + p_1  n_3)} \right) e^{ip_3 n_1} \frac{q_3^{n_2+n_3}}{q_2^{n_1}} S^{3,+}_{n_1} S^{2,+}_{n_2} S^{2,+}_{n_3} \\
	&+\left. \left( \psi_2 ((23)) e^{i(p_1  n_1 + p_2  n_3)} + \psi_2 ((123)) e^{i(p_2  n_1 + p_1  n_3)} \right) e^{ip_3 n_2} \frac{q_3^{n_1+n_3}}{q_2^{n_2}} S^{2,+}_{n_1} S^{3,+}_{n_2} S^{2,+}_{n_3} \right] |0\rangle \\
	&\approx \sum_{\substack{n_1 < n_2 < n_3 \\ 1\leq n_i \leq L}}^L \left[ (n_1 - n_2) \frac{\tilde{u}_- - \tilde{u}'_-}{u_-}\frac{(\xi_3 u_-)^{n_1 + n_2}}{(-\xi_2 u_+)^{n_3}} S^{2,+}_{n_1} S^{2,+}_{n_2} S^{3,+}_{n_3}+ (n_2 - n_3) \frac{\tilde{u}_- - \tilde{u}'_-}{u_-}  \right. \\
	&\left.\cdot\frac{(\xi_3 u_-)^{n_2 + n_3}}{(-\xi_2 u_+)^{n_1+L}} S^{3,+}_{n_1} S^{2,+}_{n_2} S^{2,+}_{n_3} + (n_1 - n_3 - L) \frac{\tilde{u}_- - \tilde{u}'_-}{u_-} \frac{(\xi_3 u_-)^{n_1 + n_3+L}}{(-\xi_2 u_+)^{n_2+L}} S^{2,+}_{n_1} S^{3,+}_{n_2} S^{2,+}_{n_3} + \right] \\
	&\epsilon^{\alpha + (\alpha-1) [(n_1 -n_3) + (n_2 -n_3)] } |0\rangle \ .
\end{align*}
We can see that the terms that dominate in this expansion are
\begin{align*}
	\frac{|\psi \rangle}{\epsilon^{3-2\alpha}} &\approx \left[ \sum_{n_1=1}^{L-2} \frac{\xi_3(\tilde{u}'_- - \tilde{u}_- )}{\xi_2^2 u_+^2}  \left( -\frac{\xi_3^2 u_-^2}{\xi_2 u_+} \right)^{n_1} S^{2,+}_{n_1} S^{2,+}_{n_1+1} S^{3,+}_{n_1+2} \right] |0\rangle + \frac{\tilde{u}'_- - \tilde{u}_- }{u_-}  \frac{(\xi_3 u_-)^{2L-1}}{(-\xi_2 u_+)^{L+1}} \\
	&\cdot S^{3,+}_{1} S^{2,+}_{L-1} S^{2,+}_{L} |0\rangle + \frac{\tilde{u}'_- - \tilde{u}_- }{u_-}  \frac{(\xi_3 u_-)^{2L+1}}{(-\xi_2 u_+)^{L+2}} S^{2,+}_{1} S^{3,+}_{2} S^{2,+}_{L} |0\rangle \ .
\end{align*}
We can combine the three terms of this expression into one, giving us
\begin{equation}
	\frac{|\psi \rangle}{\epsilon^{3-2\alpha}} \approx \frac{\xi_3(\tilde{u}'_- - \tilde{u}_- )}{\xi_2^2 u_+^2} \sum_{n_1=1}^{L} \left( -\frac{\xi_3^2 u_-^2}{\xi_2 u_+} \right)^{n_1} S^{2,+}_{n_1} S^{2,+}_{n_1+1} S^{3,+}_{n_1+2} |0\rangle \ ,
\end{equation}
where the position of the operators $S^{i,+}_n$ should be understood modulo $L$. We can see that this is indeed the locked state studied in \cite{StaudacherAhn}. Notice that this eigenvector depends on the individual momenta separately only through the total momenta $p_1+p_2+p_3$.

The process of computing the generalised eigenvector of rank 2 is similar to the case of $M=2$. As two states $|\psi (p_1 , p_2, p_3) \rangle$ and $|\psi (p'_1 , p'_2, p'_3) \rangle$ whose momenta fulfil the relation $p_1+p_2+p_3=p'_1 +p'_2+p'_3$ will coalesce to the same eigenvector of the defective matrix, the information about the generalised eigenvector of rank 2 should be codified in the linear combination $e^{i (2p_1 + p_2)} |\psi (p_1 , p_2, p_3) \rangle - e^{i (2p'_1 + p'_2)} |\psi (p'_1 , p'_2, p'_3) \rangle$. After some algebra, we can check that
\begin{equation}
	\frac{e^{i (2p_1 + p_2)} |\psi (p_1 , p_2, p_3) \rangle - e^{i (2p'_1 + p'_2)} |\psi (p'_1 , p'_2, p'_3) \rangle}{\epsilon^{4}} \approx \text{const.} \sum_{n_1=1}^{L} \left( -\frac{\xi_3^2 u_-^2}{\xi_2 u_+} \right)^{n_1} S^{2,+}_{n_1} S^{2,+}_{n_1+2} S^{3,+}_{n_1+3} |0\rangle \ .
\end{equation}

The differences with the $M=2$ case start to appear when we move to the generalised eigenvectors of rank 3. After some algebra, we find that the limit of the appropriate linear combination of eigenvectors of the diagonalisable Hamiltonian can be written as the linear combination of two vectors that depend only on the total momentum
\begin{equation}
	a_1 \left[ \sum_{n_1=1}^{L} \left( -\frac{\xi_3^2 u_-^2}{\xi_2 u_+} \right)^{n_1} S^{2,+}_{n_1} S^{2,+}_{n_1+3} S^{3,+}_{n_1+4} |0\rangle \right] + a_2 \left[ \sum_{n_1=1}^{L} \left( -\frac{\xi_3^2 u_-^2}{\xi_2 u_+} \right)^{n_1} S^{2,+}_{n_1} S^{2,+}_{n_1+1} S^{3,+}_{n_1+3} |0\rangle \right] \ ,
\end{equation}
where $a_1$ and $a_2$ are coefficients that depend on the individual momenta, not only on the total momentum. As we have seen in section~\ref{nondistinguishableJB}, this is an indication that we have found another true eigenvector of the defective Hamiltonian (up to linear combinations with the other generalised eigenvectors we have found).

We find a similar or ever more complex situation as we explore linear combinations that contain information about generalised eigenvectors of higher rank. Thus, we would like to find a systematic method to know the number of independent vectors we will find at a given step. The answer to this question is actually not that difficult. If we analyse again the dependence on $\epsilon$ of each term in $|\psi \rangle$, we see that what is actually controlling which vectors appear at a given level is the momentum factor
\begin{equation}
	e^{i (p_1 n_1 + p_2 n_2 + p_3 n_3)} \frac{q_3^{n_1} q_3^{n_2}}{q_2^{n_3}} \sim \epsilon^{(\alpha-1) [(n_1 -n_3) + (n_2 -n_3)] } \ , \label{pre-diophantus}
\end{equation}
meaning that the only relevant contribution to how the wavefunction behaves as $\epsilon$ goes to zero depends on the separation of the excitations. Consequently, the number of independent vectors that our algorithm provides at a given step has to be equal to the number of solutions to the equation $(n_3 -n_1) + (n_3 -n_2)=n$ for the appropriate value of $n$. Notice that, as our positions $n_i$ take positive integer values, what we have to solve is a \emph{Diophantine equation}. However, this is not enough, as we have to take into account that the wavefunction ~(\ref{wavefunctiongeneralM}) is written with the restriction $1\leq n_1 < n_2 <n_3 \leq L$. This means that we have to supplement the Diophantine equation with the restrictions $1\leq (n_3 - n_2) < (n_3 - n_1) \leq L-1$.

Both the Diophantine equation and the restrictions simplify significantly if we perform the change of coordinates
\begin{equation}
	(n_3 - n_2) = x_2 + 1 \ , \qquad (n_3 - n_1)=x_1+x_2+2 \ ,
\end{equation}
giving us
\begin{equation}
	x_1+2x_2=n-3=\Delta \ ,  \qquad x_1 \geq 0 \ , \qquad x_2 \geq 0 \ , \qquad x_1+x_2+2 \leq L-1 \ .
\end{equation}
We can further simplify these conditions by introducing an additional variable $x_0\geq 0$ to rewrite the last inequality as an equality. Thus, the set of equations we have to solve is
\begin{equation}
	\left\{\begin{matrix}
	x_0 + x_1 + x_2=L-3 \\
	x_1 + 2x_2 = \Delta \\
	x_j \geq 0
	\end{matrix} \right. \label{Diophantine} \ .
\end{equation}
After all these rewritings, our equations take the standard form of a system of linear Diophantine equations.

\begin{table}[t]
\begin{center}
	\begin{tabular}{|c|c|c|} \hline
	$\Delta$ & $(x_0 , x_1 , x_2)$ & Number of sol. \\ \hline
	0 & $(3,0,0)$ & 1 \\ \hline
	1 & $(2,1,0)$ & 1 \\ \hline
	2 & $(1,2,0)$, $(2,0,1)$ & 2 \\ \hline
	3 & $(0,3,0)$, $(1,1,1)$ & 2 \\ \hline
	4 & $(0,2,1)$, $(1,0,2)$ & 2 \\ \hline
	5 & $(0,1,2)$ & 1 \\ \hline
	6 & $(0,0,3)$ & 1 \\ \hline
	\end{tabular}
	\qquad
	\begin{tabular}{|c|c|c|} \hline
	$\Delta$ & $(x_0 , x_1 , x_2)$ & Number of sol. \\ \hline
	0 & $(4,0,0)$ & 1 \\ \hline
	1 & $(3,1,0)$ & 1 \\ \hline
	2 & $(2,2,0)$, $(3,0,1)$ & 2 \\ \hline
	3 & $(1,3,0)$, $(2,1,1)$ & 2 \\ \hline
	4 & $(0,4,0)$, $(1,2,1)$, $(2,0,2)$ & 3 \\ \hline
	5 & $(0,3,1)$, $(1,1,2)$ & 2 \\ \hline
	6 & $(0,2,2)$, $(1,0,3)$ & 2 \\ \hline
	7 & $(0,1,3)$ & 1 \\ \hline
	8 & $(0,0,4)$ & 1 \\ \hline
	\end{tabular}
\end{center}
	\caption{Solutions to the system of Diophantine equations~(\ref{Diophantine}) for the cases of $L=6$ and $L=7$. In both cases we consider $M=3$ and $K=1$.} \label{Solutions}
\end{table}

As examples, we have collected the solutions for the cases of $L=6$ and $L=7$ in Table~\ref{Solutions}. For the case of $L=6$, the system of equations has one solution or more for seven different values of $\Delta$, while it has two solutions for three different values of $\Delta$. This allows us to conjecture that the Jordan normal form of the Hamiltonian contains a Jordan cell of size 3 and a Jordan cell of size 7. Similarly, for the case of $L=7$ we get one solution or more for nine different values of $\Delta$, we get two solutions or more for five different values of $\Delta$, and we get three solutions just once. Thus, we can conjecture that the Jordan chains for this case are of size one, five and seven. We can see that our results match with the ones in \cite{StaudacherAhn}.

We will now see how to generalise these results to any length $L$. For that, we should notice that $(x_0+1, x_1, x_2)$ is a solution for length $L+1$ if $(x_0, x_1, x_2)$ is a solution for length $L$. Therefore, we can find the number of solutions recursively by adding the solutions associated with $x_0=0$. In order to know how many solutions we have to add and how we have to add them, we have to solve the system of Diophantine equations~(\ref{Diophantine}) with $x_0=0$
\begin{equation}
	\left\{\begin{matrix}
	x_1 + x_2=L-3 \\
	x_1 + 2x_2 = \Delta \\
	x_j \geq 0
	\end{matrix} \right. \label{reducedDiophantine} \ .
\end{equation}
In contrast with our previous system of equations, this one has exactly one solution for a given value of $\Delta$, characterised by $x_2=\Delta+3-L$ and $x_1=2L-6-\Delta$. However, not all these solutions are admissible due to the restrictions $x_i \geq 0$. This forces $\Delta$ to be in the interval $L-3 \leq \Delta \leq 2L-6$. With a little bit of effort, we can prove by induction that there are a total of $\lfloor \frac{\Delta}{2} \rfloor +1$ solutions if $0\leq \Delta \leq L-3$ and $\lfloor L-3-\frac{\Delta}{2} \rfloor +1$ if $L-3\leq \Delta \leq 2L-6$, where $\lfloor x \rfloor$ is the floor function, i.e., the function that gives us the largest integer less or equal to $x$. We will prove below that this result can be written more compactly in terms of a special function.

Once we have this result, we can easily compute the size of the Jordan chains for any length $L$. The longest Jordan chain will correspond to how many times the system of equations has one or more solutions. We can see that there is a total of $2L-6+1=2L-5$ different values of $\Delta$ for which equation~(\ref{Diophantine}) has one or more solutions. Due to the form of the floor function, the system of equations has two or more solutions for a total of $2L-5-4=2L-9$ different values of $\Delta$, it has three or more solutions for a total of $2L-13$ different values of $\Delta$ and so on. To sum up, the Jordan chains for the case with $M=3$ and $K=1$ have size $2L-5-4n$, with $n\geq 0$. This result is in complete agreement with the computations performed in \cite{StaudacherAhn} and  \cite{Ahn:2021emp}.

\subsubsection{General values of $M$} \label{Gaussianpolynomialssection}

A similar argument can be laid down for any value of $M$: as we approach $\epsilon\rightarrow 0$, the momentum contribution appearing in the wavefunction behaves as
\begin{equation}
	e^{i \sum_i p_i n_i} \frac{\prod_i q_3^{n_i}}{q_2^{n_M} q_2^{n_M}} \sim \epsilon^{(\alpha-1) (\sum_i n_i) - (\alpha + \beta -2) n_M}=\epsilon^{(\alpha-1) \sum_i (n_i - n_M) } \ , \label{pre-diophantus2}
\end{equation}
where we have used that $(M-K) (\alpha-1)=K (\beta -1)$ to make the final rewriting.

Now we proceed similarly to how we made the counting in the case of $M=3$. In order to count how may generalised eigenvector we find at a given step of our computation, we have to find how many different solutions the linear equation $\sum_i (n_M - n_i)=n$ has for a given value of $n$, taking into account the restriction $1\leq (n_M - n_{M-1}) < \dots < (n_M - n_2) < (n_M - n_1) \leq L-1$. If we perform the change of variables
\begin{equation}
	(n_M - n_i) = (M-i) + \sum_{j=i}^{M-1} x_j \ ,
\end{equation}
and introduce the variable $x_0 \geq 0$ such that we can rewrite the condition $(n_M - n_1) \leq L-1$ as the equality $(x_0 + n_M - n_1) = L-1$, we can again map the problem into a system of linear Diophantine equations
\begin{equation}
	\left\{\begin{matrix}
	\sum_{j=0}^{M-1} x_j=L-M \\
	\sum_{j=0}^{M-1} (j x_j ) = \Delta \\
	x_j \geq 0
	\end{matrix} \right. \label{DiophantineM} \ .
\end{equation}

\begin{table}[t]
\begin{center}
	\begin{tabular}{|c|c|c|} \hline
	$\Delta$ & $(x_0 , x_1 , x_2, x_3)$ & Number of sol. \\ \hline
	0 & $(4,0,0,0)$ & 1 \\ \hline
	1 & $(3,1,0,0)$ & 1 \\ \hline
	2 & $(2,2,0,0)$, $(3,0,1,0)$ & 2 \\ \hline
	3 & $(1,3,0,0)$, $(2,1,1,0)$, $(3,0,0,1)$ & 3 \\ \hline
	4 & $(0,4,0,0)$, $(1,2,1,0)$, $(2,0,2,0)$, $(2,1,0,1)$ & 4 \\ \hline
	5 & $(0,3,1,0)$, $(1,2,0,1)$, $(1,1,2,0)$, $(2,0,1,1)$ & 4 \\ \hline
	6 & $(0,2,2,0)$, $(0,3,0,1)$, $(1,0,3,0)$, $(1,1,1,1)$, $(2,0,0,2)$ & 5 \\ \hline
	7 & $(0,2,1,1)$, $(0,1,3,0)$, $(1,0,2,1)$, $(1,1,0,2)$ & 4 \\ \hline
	8 & $(0,1,2,1)$, $(0,0,4,0)$, $(1,0,1,2)$, $(0,2,0,2)$ & 4 \\ \hline
	9 & $(0,1,1,2)$, $(0,0,3,1)$, $(1,0,0,3)$ & 3 \\ \hline
	10 & $(0,1,0,3)$, $(0,0,2,2)$ & 2 \\ \hline
	11 & $(0,0,1,3)$ & 1 \\ \hline
	12 & $(0,0,0,4)$ & 1 \\ \hline
	\end{tabular}
\end{center}
	\caption{Solutions to the system of Diophantine equations~(\ref{DiophantineM}) for the case of $L=8$,  $M=4$ and $K=1$.} \label{Solutions2}
\end{table}

As an example, we have collected the solutions for the case of $L=8$, $M=4$ and $K=1$ in Table~\ref{Solutions2}. Following a similar logic as in the $M=3$, we can see that the Jordan chains in this example have lengths $13$, $9$, $7$, $5$ and $1$. This matches the results from Table 2 in \cite{StaudacherAhn}.

Let us focus now on solving the Diophantine equation for general values of $L$. For any value of $M$ we can prove that $(x_0+1, x_1, \dots ,x_{M-1})$ is a solution for length $L+1$ if $(x_0, x_1, \dots ,x_{M-1})$ is a solution for length $L$. Therefore, we can find the number of solutions recursively by adding the solutions associated with $x_0=0$. In order to know how many solutions we have to add and where we have to add them, we have to solve the system of Diophantine equations~(\ref{DiophantineM}) with $x_0=0$. Substituting $x_0=0$ in those equations and subtracting the first equation form the second one once, we find
\begin{equation}
	\left\{\begin{matrix}
	\sum_{j=1}^{M-1} x_j=L-M \\
	\sum_{j=1}^{M-1} [(j-1) x_j ] = \Delta +M-L \\
	x_j \geq 0
	\end{matrix} \right. \label{reducedDiophantineM} \ ,
\end{equation}
which is exactly the Diophantine equation for the case with length $L-1$, $M-1$ total excitations and $K=1$. Thus, we can find the number of solutions of the system of Diophantine equations~(\ref{DiophantineM}) for given values of $\Delta$, $L$ and $M$ by taking the number solution for $\Delta$, $L-1$ and $M$ and adding the number of solutions for $\Delta+M-L$, $L-1$ and $M-1$. We have presented an explicit example of how this computation works in Table~\ref{Solutions3}.

From this table, we can read that the Jordan chains for the case of $L=9$, $M=4$ and $K=1$ have lengths $4$, $6$, $8$, $10$, $12$ and $16$; while for the case of $L=10$, $M=4$ and $K=1$ they have lengths $3$, $7$, $7$, $9$, $11$, $13$, $15$ and $19$. These numbers match precisely the computations from \cite{StaudacherAhn}.

For given values of $L$ and $M$, we can construct the following generating function
\begin{equation}
	F(L,M,x)=\sum_{\Delta=0}^\infty x^\Delta \# \{\text{Solutions of eq.~\ref{DiophantineM} for given values of } L,M,\Delta\} \ .
\end{equation}
In terms of this generating function, the recurrence relation we have found can be written as
\begin{equation}
	F(L,M,x)=F(L-1,M,x)+x^{L-M} F(L-1,M-1,x) \ ,
\end{equation}
which has to be supplemented with the initial condition $F(L,2,x)=\sum_{j=0}^{L-2} x^j$. This condition comes from our previous result for $M=2$, where we found a single Jordan cell of size $L-1$ for each allowed value of the total momentum.

The solution to this recurrence relation is a known family of functions called \emph{Gaussian polynomials} or \emph{$q$-deformed binomial coefficients}. In particular
\begin{equation}
	F(L,M,x)=\binom{L-1}{M-1}_x \ .
\end{equation}
In fact, we can check that the rows of Table~\ref{Solutions3} coincide with the coefficients of the polynomials $\binom{8-1}{4-1}_x$, $x^{9-4}\binom{8-1}{3-1}_x$, $\binom{9-1}{4-1}_x$, $x^{10-4}\binom{9-1}{3-1}_x$ and $\binom{10-1}{4-1}_x$ respectively.\footnote{We want to thank the online database ``On-Line Encyclopedia of Integer Sequences'', as otherwise we would have not been able to identify the sequences of number we got from the system of Diophantine equations with the coefficients of Gaussian polynomials.}

This generating function was computed independently without using integrability in \cite{Ahn:2021emp}. The fact that we used completely different methods but we obtained matching results strengthens their validity.

\begin{table}[t]
\begin{center}
	\begin{tabular}{|c|c|c|c|c|c|c|c|c|c|c|c|c|c|c|c|c|c|c|c|c|} \hline
		\backslashbox{$(L,M)$}{$\Delta$} & 0 & 1 & 2 & 3 & 4 & 5 & 6 & 7 & 8 & 9 & 10 & 11 & 12 & 13 & 14 & 15 & 16 & 17 & 18 \\ \hline
		$(8,4)$ & 1 & 1 & 2 & 3 & 4 & 4 & 5 & 4 & 4 & 3 & 2 & 1 & 1 & 0 & 0 & 0 & 0 & 0 & 0 \\ \hline
		$(8,3) \cdot q^{8-3}$ & 0 & 0 & 0 & 0 & 0 & 1 & 1 & 2 & 2 & 3 & 3 & 3 & 2 & 2 & 1 & 1 & 0 & 0 & 0 \\ \hline
		$(9,4)$ & 1 & 1 & 2 & 3 & 4 & 5 & 6 & 6 & 6 & 6 & 5 & 4 & 3 & 2 & 1 & 1 & 0 & 0 & 0 \\ \hline
		$(9,3) \cdot q^{9-3}$ & 0 & 0 & 0 & 0 & 0 & 0 & 1 & 1 & 2 & 2 & 3 & 3 & 4 & 3 & 3 & 2 & 2 & 1 & 1 \\ \hline
		$(10,4)$ & 1 & 1 & 2 & 3 & 4 & 5 & 7 & 7 & 8 & 8 & 8 & 7 & 7 & 5 & 4 & 3 & 2 & 1 & 1 \\ \hline
	\end{tabular}
\end{center}
	\caption{Explicit computation of the number of solutions for the cases $L=9$, $M=4$ and $K=1$ and $L=10$, $M=4$ and $K=1$, starting from the number of solutions for the case with $L=8$, $M=4$ and $K=1$.} \label{Solutions3}
\end{table}

\subsubsection{Properties of Gaussian polynomials and their consequences}

Here we plan to analyse some properties of the Gaussian polynomials and infer some features of the Jordan cells of the eclectic spin chain from them.

First, we should give a proper definition of the Gaussian polynomials
\begin{equation}
	\binom{m}{r}_q= \frac{(1-q^m) (1-q^{m-1}) \ldots (1-q^{m-r+1})}{(1-q) (1-q^2) \ldots (1-q^r)} \ .
\end{equation}
Despite its appearance, the numerator is always divisible by the denominator, so they are polynomials. From this definition, we can check that Gaussian polynomials fulfil the following properties
\begin{enumerate}
	\item Reflection: $\binom{m}{r}_q= \binom{m}{m-r}_q $.
	\item First Pascal identity: $\binom{m}{r}_q = q^r \binom{m-1}{r}_q + \binom{m-1}{r-1}_q$.
	\item Second Pascal identity: $\binom{m}{r}_q = \binom{m-1}{r}_q + q^{m-r} \binom{m-1}{r-1}_q$.
\end{enumerate}

Straightforward from the definition, we have the following identity
\begin{equation}
	\binom{m}{1}_q= \sum_{j=0}^{m-1} q^j \ .
\end{equation}
The coefficients of this Gaussian polynomial give us the information about the Jordan cells of the case of length $L=m+1$ and two excitations with different flavour, i.e. the sector $M=2$ and $K=1$ we analysed in section \ref{M2section}. As the polynomial goes up to $q^{m-1}$ and all coefficients are 1, this means that we have a single Jordan cell of size $L-1$ for each allowed value of the total momentum.

We can combine this result with the first Pascal identity to show that
\begin{equation}
	\text{Coeff}\left( \binom{m}{r}_q ; q^0 \right)= 1 \ ,
\end{equation}
where Coeff$\left( P(q) ; q^n \right)$ means the coefficient accompanying $q^n$ in the polynomial $P(q)$. In addition, we can prove that the highest non-vanishing coefficient has to be of the form
\begin{equation}
	\text{Coeff}\left( \binom{m}{r}_q ; q^{(m-r)r} \right)= 1 \ .
\end{equation}
This means that the largest Jordan cell for the case of length $L$, $M$ excitations and $K=1$ has size $(L-M)(M-1)+1$. We can check that this result perfectly reproduces the values $L-1$ and $2L-5$ we found for $M=2$ and $M=3$ respectively.

We can again follow the procedure explained in section \ref{distinguishableJB} and show that this Jordan chain is not spoiled by chain mixing. This is quicker to prove if we set $\xi_1=\xi_2=0$, as $\mathbf{\hat{H}}_{0 , 0 , \xi_3}^\dagger$ can only move the 2's in the state to the left until they find a 3. As we have a total of $L-M$ 1's, we can apply $\mathbf{\hat{H}}_{0 , 0 , \xi_3}^\dagger$ a total of $(M-1)(L-M)$ times, which implies that the Jordan chain has size $(L-M)(M-1)+1$. This proves the conjecture regarding chain mixing from \cite{Ahn:2021emp} for the longest Jordan chain.

Although it takes a bit more effort, we can prove the following properties
\begin{align}
	\text{Coeff}\left( \binom{m}{r}_q ; q^1 \right) &= 1 \ , & \text{Coeff}\left( \binom{m}{r}_q ; q^{(m-r)r-1} \right) &= 1 \ ,
\end{align}
which holds for general $r$ and $m>r$, and
\begin{align}
	\text{Coeff}\left( \binom{m}{r}_q ; q^2 \right) &= 2 \ , & \text{Coeff}\left( \binom{m}{r}_q ; q^{(m-r)r-2} \right) &= 2 \ ,
\end{align}
which holds for general $r\geq 2$ and $m>$Max$(r,3)$. This means that, for a fixed value of the total momentum, we have a second Jordan cell only if $M\geq 3$, and it has size $(L-M)(M-1)-3$. This result agrees with the value $2L-9$ we found for $M=3$.

Another interesting property we can prove is that
\begin{equation}
	\text{Coeff}\left( \binom{m}{r}_q ; q^s \right)= \text{Coeff}\left( \binom{m}{r}_q ; q^{(m-r)r-s} \right) \ ,
\end{equation}
that is, the Gaussian polynomials are palindromic polynomials. This means that the size of all Jordan cells have to differ by even numbers.

\section{Conclusions}

In this article, we have presented an algorithm to construct the generalised eigenvectors associated with a defective Hamiltonian by studying how a perturbed version of that Hamiltonian that is diagonalisable approaches its exceptional point, that is, the point where it becomes non-diagonalisable. The idea behind this method is that we need to consider the limit of the linear combination of at least $n$ eigenvectors of the diagonalisable matrix to find the generalised eigenvector of rank $n$ of the defective matrix.

We have proven that this method is complete for the case when all the Jordan blocks of the Hamiltonian have different eigenvalues. By this statement, we mean that we are able to find all the generalised eigenvectors of the Hamiltonian, to which eigenvalue they are associated, and what their rank is. In addition, we have shown that this method can be applied to the case where there are two or more Jordan blocks that share the same eigenvalue. Although we have proven that the method also works in this case, in the sense that it retrieves all the generalised eigenvectors and tells us to which eigenvalue they are associated, there are many occasions in which the method is not able to correctly identify the rank of the eigenvector. This is due to the chain mixing effect, as some generalised eigenvectors may appear at the end of a different Jordan chain, so we may misguidedly consider Jordan chains that are longer or shorter than they really are (even up to the point where two Jordan chains appear as just one). Thankfully, our method is capable of identifying correctly some of the generalised eigenvectors of rank 1. We have shown that this is enough information to reconstruct the full Jordan chain, allowing us to disentangle the chain mixing.

We applied this method to the case of the eclectic spin chain introduced in \cite{Ipsen:2018fmu} to describe the one-loop dilatation operator of the fishnet conformal field theory. In that article and its continuation, \cite{StaudacherAhn}, the authors described in detail the Hamiltonian of this spin chain and how the algebraic Bethe ansatz works on that setting. In particular, they studied the Bethe roots and Bethe states of this model by computing the strongly twisted limit of the Bethe roots and Bethe states constructed for the theory at finite twist. Using this approach, they were able to find the correct number of Bethe roots, but due to the coalescence of eigenstates at the exceptional point, there were not able to find the full set of Bethe states.

Using the method we presented here, that encourages us to compute limits of linear combinations of eigenvectors that coalesce to the same vector at the exceptional point, we were able to find the generalised eigenvectors of this eclectic spin chain, together with the eigenvectors that the authors of \cite{StaudacherAhn} knew that existed from numerical computations but could not find by computing the limit of eigenvectors. In fact, our counting of eigenstates and generalised eigenstates is in agreement with the counting they do for the particular values of the number of excitations we consider here.

The counting from \cite{StaudacherAhn} was later refined to any values of $M$ and $K$ in \cite{Ahn:2021emp} by means of combinatorial arguments. The cornerstone of this computation is the definition of generalised eigenvector and the fact that the highest eigenstate of the longest Jordan chain is always an anti-locked state. This allows the authors to use orthogonality and combinatorial arguments to extract all the generalised eigenvectors with no knowledge of the finitely twisted spin chain. All the results we presented here are consistent with the ones they obtain for $K=1$. In fact, we want to stress that they also find that the number of generalised eigenstates of $\mathbf{\hat{H}}_{(\xi_1,\xi_2,\xi_3)}$ can be codified in terms of the same Gaussian polynomials. As their method and our are substantially different, this matching is more than welcome.

Despite our success in finding and classifying the generalised eigenstates of the eclectic spin chain for $K=1$, we have only showed that chain mixing is not present for the longest of the Jordan chains. This is because, although we can reconstruct a Jordan chain from its lowest generalised eigenvector, the process is tedious and becomes more cumbersome as the chain gets longer (as we are reconstructing its left generalised eigenvectors one by one). This partly confirms the conjecture laid down in \cite{Ahn:2021emp} claiming that there is no chain mixing in the eclectic spin chain.

Another interesting direction we would like to pursue is to investigate if we can adapt the method to the Algebraic Bethe Ansatz. Among other advantages, the Nested Algebraic Bethe Ansatz gives us a systematic method for constructing the eigenstates of a diagonalisable integrable system, which allows us to describe them and their properties without too much effort. Despite that, the explicit form of the wavefunctions is not immediately available, and requires us to make some additional algebraic computations that most of the time are far from easy. This is the main reason why we made use of the NCBA in this article, as our method requires us to know the wavefunctions of the diagonalisable Hamiltonian in detail.

In fact, the discussion at the end of section~\ref{nondistinguishableJB} gives us a hint of the form that generalised Bethe vectors should have. If $\mathbb{B} |0\rangle$ is an eigenvector associated with the non-diagonalisable transfer matrix $\tau (u)$ with eigenvalue $\Lambda (u)$, then
\begin{equation}
	[\tau (u) - \Lambda (u)]^l [\tau (u)^\dagger - \Lambda (u)^* ]^m \mathbb{B} |0\rangle \ ,
\end{equation}
has to be a generalised eigenvector of $\tau (u)$ if $[\tau (u)^\dagger - \Lambda (u)^* ]^m \mathbb{B} |0\rangle\neq 0$ and $[\tau (u)^\dagger - \Lambda (u)^* ]^{m+1} \mathbb{B} |0\rangle=0$. Although we may be able to construct eigenvectors of the non-diagonalisable transfer matrix by computing the limit of eigenvectors of a diagonalisable one, we have shown that not every eigenvector can be constructed using this procedure due to the presence of chain mixing (see also \cite{StaudacherAhn}). This means that the limit of the $B$ operators of the diagonalisable transfer matrix are not enough for an Algebraic Bethe Ansatz for non-diagonalisable $R$-matrices.  Nevertheless, the above argument would work for every eigenvector regardless of its origin. Thus, the only obstacle would be to find a method to construct all the eigenvectors of the non-diagonalisable transfer matrix $\tau (u)$ from first principle.

Finally, in this article we have applied this algorithm to compute generalised eigenvectors of the eclectic spin chain, but nothing prevents us from applying it to other non-diagonalisable integrable systems. All the integrable systems with nearest-neighbour interaction and either $\mathfrak{su}(2)$ or $\mathfrak{su}(2) \oplus \mathfrak{su}(2)$ symmetry were classified in \cite{DeLeeuw:2019gxe} and \cite{DeLeeuw:2019fdv,deLeeuw:2020ahe,deLeeuw:2020xrw} respectively. Some of the R-matrices that appear in these classifications are non-diagonalisable, making their transfer matrices good candidates for the algorithm presented here.

\section{Acknowledgements}

We are very thankful to Changrim Ahn, Luke Corcoran, Matthias Staudacher, Alessandro Torrielli, Roberto Ruiz and Rafael Hernández for helpful discussions. We are grateful to Luke Corcoran, Matthias Staudacher, Alessandro Torrielli and Roberto Ruiz for reading the manuscript and providing very useful comments. LW is funded by a University of Surrey Doctoral College Studentship Award. This work is supported by the EPSRC-SFI grant EP/S020888/1 {\it Solving Spins and Strings}. 

No data beyond those presented and cited in this work are needed to validate this study.



\end{document}